\pdfoutput=1 
\documentclass[a4paper,english,cleveref,autoref,thm-restate]{lipics-v2019}

\usepackage{physics}
\usepackage{amsmath}
\usepackage{tikz}
\usepackage{mathdots}
\usepackage{yhmath}
\usepackage{cancel}
\usepackage{color}
\usepackage{siunitx}
\usepackage{array}
\usepackage{multirow}
\usepackage{amssymb}
\usepackage{gensymb}
\usepackage{tabularx}
\usepackage{booktabs}
\usetikzlibrary{fadings}
\usetikzlibrary{patterns}
\usetikzlibrary{shadows.blur}
\usetikzlibrary{shapes}

\newcommand{\eq}{\textsc{equality}}
\newcommand{\hide}[1]{ }
\newcommand{\acc}{\mathsf{acc}}

\newcommand{\symmetry}{\textsc{symmetry}}
\newcommand{\amos}{\textsc{amos}}
\newcommand{\twocoleq}{\textsc{2-col-eq}}

\newcommand{\id}{\ensuremath{\mathsf{id}}}

\newcommand{\prot}{\mathsf{prot}}
\newcommand{\M}{\mbox{\rm\textsf{M}}}
\newcommand{\MA}{\mbox{\rm\textsf{MA}}}
\newcommand{\A}{\mbox{\rm\textsf{A}}}
\newcommand{\dM}{\mbox{\rm\textsf{dM}}}
\newcommand{\dAM}{\mbox{\rm\textsf{dAM}}}
\newcommand{\dMA}{\mbox{\rm\textsf{dMA}}}
\newcommand{\dMAM}{\mbox{\rm\textsf{dMAM}}}
\newcommand{\dAMA}{\mbox{\rm\textsf{dAMA}}}
\newcommand{\dAMAM}{\mbox{\rm\textsf{dAMAM}}}
\newcommand{\cL}{\mathcal{L}}
 \newcommand{\completeness}{\ensuremath{\mathbf{Completeness: }}}
 \newcommand{\soundness}{\ensuremath{\mathbf{Soundness: }}} 
\newcommand{\lipics}[1]{\textcolor{lipicsGray}{\textbf{\textsf{#1}}}}

\newcommand{\sym}{\mathrm{sym}}
\newcommand{\priv}{\mathrm{p}}
\newcommand{\pub}{\mathrm{s}}

\newcommand{\poly}{\mathrm{poly}}
\newcommand{\eps}{\varepsilon}

\newcommand{\cO}{\mathcal{O}}

\title{Shared v\/s Private Randomness in Distributed Interactive Proofs}

\titlerunning{Randomness in Distributed Interactive Proofs} 

\author{Pedro Montealegre\footnote{Corresponding author}}{Facultad de Ingenier\'{\i}a y Ciencias, Universidad Adolfo Ib\'a\~nez, Chile}{p.montealegre@uai.cl}{}{}
\author{Diego Ram\'{\i}rez-Romero}{Departamento de Ingenier\'{\i}a Matem\'atica, Universidad de Chile, Chile}{dramirez@dim.uchile.cl}{}{}
\author{Ivan Rapaport}{DIM-CMM (UMI 2807 CNRS), Universidad de Chile, Chile}{rapaport@dim.uchile.cl}{}{}

\authorrunning{P. Montealegre, D. Ram\'{\i}rez and I. Rapaport} 

\Copyright{P. Montealegre, D. Ram\'{\i}rez and I. Rapaport} 



\funding{Partially supported by CONICYT via PIA/ Apoyo a Centros Cient\'{\i}ficos y Tecnol\'ogicos de Excelencia AFB 170001 (P.M. and I.R.), FONDECYT 1170021 (D.R. and I.R.) and FONDECYT 11190482 (P.M.) and PAI + Convocatoria Nacional Subvenci\'on a la Incorporaci\'on en la Academia A\~no 2017 + PAI77170068 (P.M.).}

\nolinenumbers 

\hideLIPIcs  

\EventEditors{}
\EventNoEds{2}
\EventLongTitle{31st International Symposium on Algorithms and Computation (ISAAC 2020)}
\EventShortTitle{ISAAC 2020}
\EventAcronym{ISAAC}
\EventYear{2020}
\EventDate{December 14--18, 2020}
\EventLocation{Hong Kong}
\EventLogo{}
\SeriesVolume{}
\ArticleNo{}

\begin{document}
\maketitle

\begin{abstract}

In distributed interactive proofs, the nodes of a graph G interact with a powerful but untrustable prover who tries to convince them, in a small number of rounds and through short messages, that G satisfies some property. This series of interactions is followed by a phase of distributed verification, which may be either deterministic or randomized, where nodes exchange messages with their neighbors.

The nature of this last verification round defines the two types of interactive protocols. We say that the protocol is of Arthur-Merlin type if the verification round is deterministic. We say that the protocol is of Merlin-Arthur type if, in the verification round, the nodes are allowed to use a fresh set of random bits.

In the original model introduced by Kol, Oshman, and Saxena [PODC 2018], the randomness was private in the sense that each node had only access to an individual source of random coins. Crescenzi, Fraigniaud, and Paz [DISC 2019] initiated the study of the impact of shared randomness (the situation where the coin tosses are visible to all nodes) in the distributed interactive model.

In this work, we continue that research line by showing that the impact of the two forms of randomness is very different depending on whether we are considering Arthur-Merlin protocols or Merlin-Arthur protocols. While private randomness gives more power to the first type of protocols, shared randomness provides more power to the second. Our results also connect shared randomness in distributed interactive proofs with distributed verification, and new lower bounds are obtained.

\end{abstract}

\section{Introduction}\label{sec:introduction}
Distributed decision refers to the task in which the nodes of a connected graph $G$ have to collectively decide whether $G$ satisfies some graph property~\cite{naor1995can}. For performing any such task,
the nodes exchange messages through the edges of $G$. The input of distributed decision problems may also include labels given to the nodes and/or to the edges of $G$. For instance, the nodes could decide whether $G$ is properly colored, or decide whether the weight of the minimum spanning tree lies below some threshold.
 
Acceptance and rejection are defined as follows. 
If $G$ satisfies the property, then all nodes must accept; otherwise, at least one node must reject~\cite{korman2010proof}.
This type of algorithms could be used in distributed fault-tolerant computing, where the nodes, with some regularity,
must check whether the current network configuration is in a legal state for some Boolean predicate~\cite{fraigniaud2019randomized}. Then, if the configuration becomes illegal at some point, the rejecting node(s) raise the alarm or launch a recovery procedure.

Deciding whether a given coloring is proper can be done locally, by exchanging messages between neighbors. These types of properties
are called {\emph{locally decidable}}. Nevertheless, some other properties, such as deciding whether $G$ is a tree, are not. 
 As a remedy, the notion of {\emph{proof-labeling scheme}} (PLS) was introduced~\cite{korman2010proof}. Similar variants were also introduced: non-deterministic local decisions~\cite{fraigniaud2013towards}, locally checkable proofs~\cite{goos2016locally}, and others.
 
 Roughly speaking, in all these models, a powerful prover gives to every node $v$ a certificate $c(v)$.
This provides $G$ with a global distributed-proof. Then, every node $v$ performs a
local verification using its local information together with $c(v)$. PLS can be seen as a distributed counterpart to the class NP,
where, thanks to nondeterminism, the power of distributed algorithms increases.

Just as it happened in the centralized framework~\cite{goldreich1991proofs,goldwasser1989knowledge},
a natural step forward is to consider a model where the nodes are allowed to have more than one interaction with the prover. 
In fact, with the rise of the Internet, prover-assisted computing models are more relevant than ever.
We can think of asymmetric applications like Facebook, where, together with the social network itself,
there is a very powerful central entity that stores a large amount of data (the topology of the network, preferences, and activities of the users, etc.). 
Or we can consider Cloud Computing, where computationally limited devices delegate costly computations to a cloud with tremendous computational power. 
The central point lies in the fact that these devices may not trust their cloud service (as it may be malicious, selfish, or buggy). Therefore, the nodes must regularly verify
the correctness of the computation performed by the cloud service.

Interestingly, there is no gain when interactions are all deterministic. When there is no randomness, the prover, from the very beginning, has all the information required to simulate the interaction with the nodes. Then, in just one round, he could simply send to each node the transcript of the whole communication, and the nodes simply verify that the transcript is indeed consistent. A completely different situation occurs when the nodes have access to some kind of randomness~\cite{baruch2015randomized,fraigniaud2019randomized}. In that case, the exact interaction with the nodes is unknown to the prover until the nodes communicate the realization of their random variables. Adding a randomized phase to the non-deterministic
phase gives more power to the model~\cite{baruch2015randomized,fraigniaud2019randomized}.

 Two model variants arise in this new randomized scenario, regarding the order of the phases.
Assume that we have two phases. When the random phase precedes the non-deterministic phase, we refer to 
{\emph{distributed Arthur-Merlin protocols}}, and we denote them by \dAM\ (following the terminology and notation of~\cite{kol2018interactive}).
Conversely, when nodes access randomness only after receiving the certificates,
we refer to {\emph{distributed Merlin-Arthur protocol}}s, and we denote them by \dMA. 
Note that Merlin is the powerful but untrustable prover of the PLS model,
while Arthur represents the nodes, which are simple and limited verifiers that can flip coins.

In a \dMA\ protocol, the prover does not see the nodes' randomness when choosing the certificates. Instead, only once the prover assigns certificates to the nodes, each node randomly selects a message that broadcasts to its neighbors. Then, each node decides whether to accept or reject, based on its randomness, input, certificate, and the messages it received from its neighbors.

These definitions can be easily extended to a more general setting~\cite{crescenzi2019trade}, where the number of interactions between Arthur and Merlin is constant but not fixed to only one interaction per player. This model was introduced in~\cite{kol2018interactive} and further studied in~\cite{crescenzi2019trade,fraigniaud2019distributed,naor2020power}. For instance, a \dMAM\ protocol involves three interactions: Merlin provides a certificate to Arthur, then Arthur queries Merlin by sending a random string. Finally, Merlin replies to Arthur's query by sending another certificate. Recall that this series of interactions is followed by a phase of distributed verification performed between every node and its neighbors. When the number of interactions is $k$ we refer to $\dAM[k]$ protocols (if the last player is Merlin) and $\dMA[k]$ protocols (otherwise). For instance, $\dAM[2] = \dAM$, $\dMA[3] = \dAMA$, etc. Also, the scenario of distributed verification, where there is no randomness and only Merlin interacts, corresponds to $\dAM[1]$, which we denote by \dM. In other words, \dM\ is the PLS model.

In distributed interactive proofs, Merlin tries to convince the nodes that $G$ satisfies some property in a small number of rounds and through short messages. We say that an algorithm uses $\cO(f(n))$ bits if the messages exchanged between the nodes (in the verification round) and also the messages exchanged between the nodes and the prover are upper bounded by $\cO(f(n))$. We include this {\emph{bandwidth bound}} in the notation, which becomes $\dMA[k,f(n)] $ and $\dAM[k,f(n)]$
for the corresponding protocols.

In this article we cope with an important issue, well-studied in the context of communication complexity, but much less 
considered in distributed computing, related to the visibility of the coins: they can be either shared or 
private~\cite{babai1997randomized,becker2014simultaneous,fischer2016public,KNR99,newman1996public}.
The theory of distributed decision has restricted itself to private randomness, 
in the sense that each node has only access to a private source of random coins. 
These coins are shared with the prover but remain private to the other nodes. We explore the role of \emph{shared randomness}, that is, the situation on which the same set of random bits is produced on every node. The issue of shared randomness in distributed interactive proofs was explicitly formulated by Naor, Parter, and Yogev~\cite{naor2020power}. It is also expressly addressed in Crescenzi, Fraigniaud, and Paz~\cite{crescenzi2019trade}.

For distinguishing the two types of randomness, we denote the private randomness setting by $\dAM^{\priv}[k,f(n)]$, and the shared randomness setting by $\dAM^{\pub}[k,f(n)]$.
Also, as explained before, we omit the number of interactions $k$ when they are 2. For instance, we denote $\dAM^{\priv}[2,f(n)]$ simply by $\dAM^{\priv}[f(n)]$.

Some distributed problems are hard, even when a powerful prover provides the nodes with certificates.
It is the case of \symmetry, the language of graphs having a non-trivial automorphism (i.e., a non-trivial one-to-one mapping from the set of nodes to itself preserving edges). Any proof labelling scheme recognizing \symmetry\ requires certificates of size $\Omega(n^2)$~\cite{goos2016locally}.

Many problems requiring $\Omega(n^2)$-bit certificates in any PLS, such as \symmetry, admit distributed interactive protocols with small certificates, and very few interactions. In fact, 
$\symmetry$ is in both $\dMAM^{\priv}[\log n]$ and $\dAM^{\priv}[n \log n]$~\cite{kol2018interactive}.
Moreover, ${\overline{\symmetry}}$ (i.e. the languages of graphs not having a non-trivial automorphism) belongs to $\dAMAM^{\priv}[\log n]$~\cite{naor2020power}.

In~\cite{crescenzi2019trade}, the authors explore the role of shared randomness in distributed interactive proofs.
They prove that private randomness does not limit the power of Arthur-Merlin
protocols compared to shared randomness, up to a small additive factor in the certificate
size. Roughly, they show that, if ${\mathcal L} \in \dAM^{\pub}[k,f(n)]$, 
then ${\mathcal L} \in \dAM^{\priv}[k,f(n)+ \log n]$. 

We deepen this study by finding explicit inclusions and separations between models.

\subsection{Our Results}

In Section \ref{sec:lim} we show that any interactive protocol using shared randomness
 can be derandomized into a non-interactive proof, with an exponential-factor overhead in the bandwidth. 
Roughly, we prove that, 
if $\mathcal{L} \in \dAM^{\pub}[k, f(n)]$, then $\mathcal{L} \in \dM(2^{O(k\:f(n))}+ \log n)$.
From this we conclude many lower bounds. For instance, we can conclude that  \textsc{symmetry} $\in \dAM^{\pub}[k,\Omega(\log n)]$, for any fixed $k$. This result is tight,
because it is already known that \textsc{symmetry} $\in \dMAM^{\pub}[\log n]$ (in fact, it is known that \textsc{symmetry} $\in \dMAM^{\priv}[\log n]$~\cite{kol2018interactive}, 
but the private coin protocol can be easily adapted to work with shared randomness).

Later, in Section \ref{sec:am}, we separate the models with private and shared randomness through the language \amos, which is the language of labeled graphs having at most one selected node. More precisely, \amos\ is the language of $n$-node graphs with labels in $\{0,1\}$, and where at most one vertex is labeled $1$. In \cite{fraigniaud2019distributed} it is shown $\amos$ is \emph{easy} for private-coin Arthur-Merlin protocols, as $\amos \in \dAM^{\priv}[1]$. We prove that $\textsc{amos} \in \dAM^{\pub}[k, \Theta(\log \log n)]$
and hence there exists an unbounded gap between the two models.

Interestingly, regarding private and shared randomness, roles are reversed when we address \dMA\ protocols
instead of \dAM\ protocols. In fact, in Section \ref{sec:ma}, we get an analogous result to that in~\cite{crescenzi2019trade} 
by proving that \dMA\ protocols with shared randomness are more powerful than \dMA\ protocols with private randomness.
More precisely, if ${\mathcal L} \in \dMA^{\priv}_\eps[f(n)]$, then  
${\mathcal L} \in \dAM^{\pub}_{\eps + \delta}[f(n) + \log n + \log (\delta^{-1})]$. We then separate the two classes. We introduce another language denoted \twocoleq, which consists of graphs with $n$-bit labels corresponding to proper 2-colorings. In other words, the language consists of bipartite graphs where
each part is colored with an $n$-bit label. We show that \twocoleq\ separates shared and private randomness on distributed Merlin-Arthur protocols. More precisely, we show first that $\twocoleq \in \dMA^{\pub}[\log n]$. Then, we show that, for $\eps<1/4$, $\twocoleq \in \dAM_{\eps}^{\priv}[\Theta(\sqrt{n})]$.

\subsection{Related Work}

The study of the role of shared and private randomness in distributed interactive was initiated very recently~\cite{crescenzi2019trade}. With respect to the case $\dAM^{\pub}[2]=\dAM^{\pub}$, the authors show that any Arthur-Merlin protocol for both \symmetry\ and ${\overline{\symmetry}}$ must have certificates and messages of size $\Omega(\log \log n)$. Note that this is stronger than just saying $\symmetry, {\overline{\symmetry}} \notin \dAM^{\pub}(o(\log\log n))$. On the positive side, in~\cite{crescenzi2019trade} the authors show that, in the $\dMA^{\pub}$ model, shared randomness helps significantly if we want to decide whether a graph has no triangles. In fact, the language of triangle-free graphs belongs to $\dMA^{\pub}[\sqrt{n}\log n]$ while any PLS requires certificates of size $n/e^{\mathcal{O}(\sqrt{\log n})}$.

By contrast, the issue of private versus shared randomness has been intensively addressed in the communication complexity framework. More precisely, in the Simultaneous Messages Model (\textsf{SM}). This two-player model was already present in Yao's seminal communication complexity paper of 1979~\cite{yao1979some}. 

In the \textsf{SM} model, the two parties are unable to communicate with each other, but, instead, can send a single message to a referee. Yao proved that the message size complexity of \textsc{Eq}, which tests whether two $n$-bit inputs are equal, is $\Theta(n)$ in the deterministic case (in fact he proved that this is also true even if players can communicate back-and-forth). Later, clear separations have been proved between deterministic, private randomness, and shared randomness algorithms. In the shared randomness setting with constant one-sided error, the message size complexity of \textsc{Eq} is $\mathcal{O}(1)$~\cite{babai1997randomized}. On the other hand, for private randomness algorithms of constant one-sided error, the message size complexity is much higher,  $\Theta(\sqrt{n})$~\cite{babai1997randomized,newman1996public}. More generally, Babai and Kimmel~\cite{babai1997randomized} proved that, for any function $f$, the use of private randomness in simultaneous messages might lead to at most a square root improvement.

There are natural ways to extend the \textsf{SM} model to more than two players. This issue is addressed in~\cite{fischer2016public} in the context of the \emph{number-in-hand} model (where each player only knows its own input, there is no input graph $G$ and players broadcast messages in each round). In problem \textsc{AllEq} there are $k$ players, each one receives a boolean vector $\{0,1\}^n$, and they have to decide whether all the $k$ vectors are equal. In problem \textsc{ExistsEq}, the $k$ players have to decide whether there exist {\emph{at least}} two players with the same input. It is not difficult to see that in both the deterministic case and the shared randomness case, the results for two players can be extended to $k$ players (the number of players is irrelevant). The private coin case is more involved than the case of shared randomness. With respect to private coin algorithms of constant error, the authors prove, for problem \textsc{AllEq}, an upper bound of $\mathcal{O}(\sqrt{n/k} + \log(\min(n,k)))$ and a lower bound of $\Omega(\log{n})$. In the case of \textsc{ExistsEq} the upper bound they show is $\mathcal{O}(\log k\sqrt{n})$ while the lower bound is $\Omega(\sqrt{n})$.

\section{Model and Definitions}\label{sec:model-def}

Let $G$ be a simple connected $n$-node graph, let $I:V(G)\to \{0,1\}^*$ be an input function assigning labels to the nodes of $G$, where the size of all inputs is polynomially bounded on $n$. Let $\id:V(G)\to\{1,\dots,\text{poly}(n)\}$ be a one-to-one function assigning identifiers to the nodes. A \emph{distributed language} $\mathcal L$ is a (Turing-decidable) collection of triples $(G,\id,I)$, called \emph{network configurations}. In this paper, we are particularly interested in two languages. The first one, denoted \amos,  is the language of graphs where at most one node is selected. The second language, denoted \twocoleq, consists in graphs with $n$-bit labels corresponding to proper 2-colorings. Formally,

\begin{itemize}
	
\item $\amos ={\Big \{} (G,\id,I) \mid I:V(G) \to \{0,1\} \mbox{ and } |\{v \in V(G): I(v)=1\}| \leq 1 {\Big \}}$,

\medskip

\item $\twocoleq ={\Big \{} (G,\id,I) \mid I:V(G) \to \{0,1\}^n \mbox{ is a proper two-coloring of $G$}{\Big \}}$.

\end{itemize}

Also, we introduce other problems that will be of interest in the following sections: \textsc{simmetry}, \textsc{diameter}, \textsc{planar}, \textsc{outerplanar}, 3-\textsc{col},  \textsc{spanning tree} and $\triangle$-\textsc{free} consisting in, respectively, deciding the existence of a non-trivial automorphism, determining whether the graph has diameter bounded by some threshold, whether the graph is (outer) planar, whether the graph is 3-colorable, whether a set of edges of the graph form a spanning tree, and whether the graph has no triangles (as subgraphs). For simplifying the notation, we denote by  $\llbracket \text{p}(x) \rrbracket$ the function that equals one iff the proposition p$(x)$ is true.

A distributed interactive protocol consists of a constant series of interactions between a {\emph{prover}} called Merlin, and a {\emph{verifier}} called Arthur. The prover Merlin is centralized, has unlimited computing power and knows the complete configuration $(G,\id,I)$. However, he can not be trusted. On the other hand, the verifier Arthur is distributed, represented by the nodes in $G$, and has limited knowledge. In fact, at each node $v$, Arthur is initially aware only of his identity $\id(v)$, and his label $I(v)$. He does not know the exact value of $n$, but he knows that there exists a constant $c$ such that $\id(v) \leq n^c$. Therefore, for instance, if one node $v$ wants to communicate his $\id(v)$ to its neighbors, then the message is of size $\cO(\log n)$.

Given any network configuration $(G, \id, I)$, the nodes of $G$ must collectively decide whether $(G, \id, I)$ belongs to some distributed language ${\mathcal L}$. If this is indeed the case, then all nodes must accept; otherwise, at least one node must reject (with certain probabilities, depending on the precise specifications we are considering).

There are two types of interactive protocols: Arthur-Merlin and Merlin-Arthur. Both types of protocols have two phases: an interactive phase and a verification phase. Let us define first {\emph{Arthur-Merlin interactive protocols}}. If Arthur is the party that starts the interactive phase, he picks a random string $r_1(v)$ at each node $v$ of $ G $ (this string could be either private or shared) and send them to Merlin. Merlin receives $r_1$, the collection of these $n$ strings, and provides every node $v$ with a certificate $c_1(v)$ that is a function of $v$, $r_1$ and $(G,\id,I)$. Then again Arthur picks a random string $r_2(v)$ at each node $v$ of $G$ and sends $r_2$ to Merlin, who, in his turn, provides every node $v$ with a certificate $c_2(v)$ that is a function of $v$, $r_1$, $r_2$ and $(G,\id,I)$. This process continues for a fixed number of rounds. If Merlin is the party that starts the interactive phase, then he provides at the beginning every node $v$ with a certificate $c_0(v)$ that is a function of $v$ and $(G,\id,I)$, and the interactive process continues as explained before. In Arthur-Merlin protocols, the process ends with Merlin. More precisely, in the last, $k$-th round, Merlin provides every node $v$ with a certificate $c_{\lceil k/2\rceil}(v)$. Then, the verification phase begins. This phase is a one-round deterministic algorithm executed at each node. More precisely, every node $v$ broadcasts a message $M_v$ to its neighbors. This message may depend on $\id(v)$, $I(v)$, all random strings generated by Arthur at $v$, and all certificates received by $v$ from Merlin. Finally, based on all the knowledge accumulated by $v$ (i.e., its identity, its input label, the generated random strings, the certificates received from Merlin, and all the messages received from its neighbors), the protocol either accepts or rejects at node $v$. Note that Merlin knows the messages each node broadcasts to its neighbors because there is no randomness in this last verification round.

A {\emph{Merlin-Arthur interactive protocols}} of $k$ interactions is an Arthur-Merlin protocol with $k-1$ interactions, but where the verification round is randomized. More precisely, Arthur is in charge of the $k$-th interaction, which includes the verification algorithm. The protocol ends when Arthur picks a 
random string $r(v)$ at every node $v$ and uses it to perform a (randomized) verification algorithm. In other words, each node $v$ randomly chooses a message $M_v$ from a distribution specified by the protocol, and broadcast $M_v$ to its neighbors. Finally, as explained before, the protocol either accepts or rejects at node $v$. Note that, in this case, Merlin does not know the messages each node broadcasts to its neighbors (because they are randomly generated). If $k=1$, a distributed Merlin-Arthur protocol is a (1-round) randomized decision algorithm; if $k=2$, it can be viewed as the non-deterministic version of randomized decision, etc.

\begin{definition}
Let ${\mathcal V}$ be a verifier and ${\mathcal M}$ a prover of a distributed interactive proof protocol for languages over graphs of $n$ nodes.
If $({\mathcal V}, {\mathcal M})$ corresponds to an Arthur-Merlin (resp. Merlin Arthur) $k$-round, $\cO(f(n))$ bandwidth protocol,
we note $({\mathcal V}, {\mathcal M}) \in {\dAM}_{\prot}[k,f(n)]$ (resp. $({\mathcal V}, {\mathcal M}) \in {\dMA}_{\prot}[k,f(n)]$).
\end{definition}

\begin{definition}
Let $\eps \leq 1/3$. The class $\dAM_{\eps}[k,f(n)]$ (resp. $\dMA_{\eps}[k,f(n)]$) is the class of languages ${\mathcal L}$ over graphs of $n$ nodes
for which there exists a verifier ${\mathcal V}$ such that, for every configuration $(G,\id,I)$ of size $n$, the
two following conditions are satisfied.

\begin{itemize}
\item \completeness
If $(G,\id,I) \in \cL$ then, there exists a prover $ {\mathcal M}$ such that 

$({\mathcal V}, {\mathcal M}) \in {\dAM}_{\prot}[k,f(n)]$
(resp. $({\mathcal V}, {\mathcal M}) \in {\dMA}_{\prot}[k,f(n)]$) and 
\noindent
\[\mathbf{Pr} \Big{[} \mathcal{V} \mbox{ accepts } (G,\id,I) \mbox{ in every node given } \mathcal{M}\Big{]} \geq 1 - \eps.\]

\item  \soundness
 If $(G,\id,I) \notin \cL$ then, for every prover $ {\mathcal M}$ such that 
 
 $({\mathcal V}, {\mathcal M}) \in {\dAM}_{\prot}[k,f(n)]$
(resp. $({\mathcal V}, {\mathcal M}) \in {\dMA}_{\prot}[k,f(n)]$),
\noindent
\[\mathbf{Pr} \Big{[} \mathcal{V} \mbox{ rejects } (G,\id,I) \mbox{ in at least one nodes given } \mathcal{M}\Big{]} \geq 1-\eps.\]
\end{itemize}
We also denote $\dAM[k,f(n)]= \dAM_{1/3}[k,f(n)]$ and $\dMA= \dMA_{1/3}[k,f(n)]$.
\end{definition}

We omit the subindex $\eps$ when its value is obvious from the context. 
For small values of $k$, instead of writing $\dAM[k,f(n)]$ and $\dMA[k,f(n)]$, we alternate \M s and \A s.
For instance: $ \dMAM[f(n)] = \dAM[3,f(n)] $, $\dAMA[f(n)] = \dMA[3,f(n)] $, etc. In particular $ \dAM[f(n)]= \dAM[2,f(n)] $, $ \dMA[f(n)] = \dMA[2,f(n)]$.

\begin{definition}
The shared randomness setting may be seen as if all the nodes, in any given round, sent the same random string to Merlin.
In order to distinguish between the settings of private randomness and shared randomness, we denote them by $\dAM^{\priv}[k,f(n)]$ and $\dAM^{\pub}[k,f(n)]$, respectively.
\end{definition}

\subsection{Simultaneous Messages Model}

In the simultaneous messages model (\textsf{SM}) there are three players, \emph{Alice}, \emph{Bob} and a \emph{referee}, who jointly want to compute a function $f(x,y)$. Alice and Bob are given inputs $x$ and $y$, respectively. The referee has no input. Alice and Bob are unable to communicate with each other, but, instead, are able to send a single message to the referee. Their messages depend on their inputs and a number of random bits. Then, using only the messages of Alice and Bob and eventually another random string, the referee has to output $f(x,y)$ (up to some error probability $\eps$, given by the coins of Alice, Bob, and the referee). A randomized protocol with error $\eps$ is \textit{correct} in the \textsf{SM} model if the answer is correct with probability at least $1-\eps$.

We are only interested in the \textsf{SM} model with \emph{private coins}, i.e. when the random strings generated by Alice, Bob, and the referee are independent. Interestingly, in this model, the power of randomness is very restricted. Indeed, in \cite{babai1997randomized}, Babai and Kimmel show that any randomized protocol computing a function $f$ in the simultaneous messages model using private coins requires messages of size at least the square root of its deterministic complexity. More precisely, if we define the deterministic complexity of $f$, $\textsf{D}(f)$,
as the size of the messages of an optimal \textsf{SM} deterministic protocol for $f$, the following proposition holds.

\begin{proposition}[\cite{babai1997randomized}, Theorem 1.4]\label{babai-kimmel}
	Let $f: X \times Y \to \{0,1\}$ be any boolean function. Let $0\leq\eps<\frac{1}{2}$. Any $\eps$-error $\mathsf{SM}$ protocol for solving 
	$f$ using private coins needs the messages to be of size
	at least $\Omega\left(\sqrt{\mathsf{D}(f)}\right)$. 
\end{proposition}

By incorporating a prover, we can define interactive proofs in the ${\mathsf{SM}}$ model. More precisely, we define $\MA^{\sym}$ as follows.

\begin{definition}
Let $f: X \times Y \to \{0,1\}$ be a boolean function. We say that $f \in {\MA}_{\eps}^{\sym}$ if there exists a protocol for Alice and Bob, where:
	
\begin{itemize}
\item A fourth player, the prover, provides Alice and Bob with a proof $m$ (which he builds as a function of the input of Alice $x \in X$ and the input of Bob $y \in Y$).
\smallskip

\item Alice and Bob simultaneously send a message to the referee, that depends on their inputs, their own randomness, and the certificate $m$ provided by the prover. Let $\omega_{x,m}(r)$ be the message sent by Alice given the input $x$ and the seed $r$ and let $\varphi_{y,m}(s)$ be the message sent by Bob, given $y$ and the seed $s$.

 \item Finally, let $\rho(\omega, \varphi)$ be the random variable indicating the referee's decision given the messages $\omega \varphi$ and its random bits.
\end{itemize}

For all $x \in X,y \in Y$, the protocol must satisfy the following: 

\begin{itemize}
\item \completeness If $f(x,y) = 1$, there exists a proof $m$ s.t. $ \mathbf{Pr}\left (\rho(\omega_{x,m}, \varphi_{y,m} = 1\right) \ge1-\eps$.

\item\soundness If $f(x,y) = 0$ then, for any proof $m$, $\mathbf{Pr} ( \rho(\omega_{x,m} , \varphi_{y,m})= 1)\leq \eps$.
\end{itemize}

\end{definition}

Let $f: X \times Y \to \{0,1\}$ be a boolean function. The cost of an $\MA^{\sym}$ protocol
that solves $f$ is the sum of the proof size, along with the maximum size of a message considering all possible random bits.
When there is no randomness we recover the classical definition of non-deterministic complexity
in the ${\mathsf{SM}}$ model, which we denote by $\M^{\sym}(f)$.

\begin{remark}
The assumption that both Alice and Bob receive the same proof does not affect the definition of the class: in case that Alice receives $m_a$ and Bob receives $m_b$ as proofs, then Merlin may concatenate $m_a m_b$ and then Alice and Bob just consider their part of the message (the referee verifies that Alice and bob received, indeed, the same message).
\end{remark}

\section{The Limits of Shared Randomness}\label{sec:lim}

In this section we show that the largest possible gap between non-interactive proofs and interactive proofs with shared randomness
is exponential.  More precisely, we show that any  interactive protocol using shared randomness
 can be derandomized into a non-interactive proof, with an exponential-factor overhead in the bandwidth. From this result we can obtain
 lower bounds, some of them even tight, for the bandwidth of interactive-proofs with shared randomness.

\begin{theorem}\label{thm:simula}
Let $k \ge 1$ and let $\mathcal{L}$ be a language such that $\mathcal{L} \in \dAM^{\pub}[k, f(n)]$. Then, $\mathcal{L} \in \dM(2^{O(k\:f(n))}+ \log n)$.
\end{theorem}

\begin{proof}
	Let $\mathcal{P}$ be a protocol deciding $\mathcal{L}$ using shared randomness, $k$ rounds of interaction, bandwidth $f(n)$, and with error probability $1/3$. We use 
	$\mathcal{P}$ to define a protocol $\mathcal{P}'$ for $\mathcal{L}$ with only one round of interaction and bandwidth $2^{\cO(k\cdot f(n))}+ \log n$. Let us fix $(G, \id, I)$, an instance of $\mathcal{L}$. 
	
	For a prover $\mathcal{M}$ for protocol $\mathcal{P}$, we define a \emph{transcript} of a node $v \in G$ as a $k$-tuple $\tau( \mathcal{M}, v) = (\tau_1, \tau_2, \dots, \tau_{k})$ such that $\tau_i \in \{0,1\}^{f(n)}$ is a sequence of bits communicated in the $i$-th round of interaction of $\mathcal{P}$, for each $i \in \{1, \dots, k\}$. If both $k$ and $i$ are even, then $\tau_i$ is a message that $\mathcal{M}$ sends to node $v$ in the $i$-th interaction. If $k$ is even and $i$ is odd, then $\tau_i$ is a random string drawn from the shared randomness. Finally, roles are reversed fin the case where $k$ is odd.
	
	Let us fix $\ell=\lfloor \frac{k}{2}\rfloor$ and let $R$ be the set of all $\ell$-tuples $r = (r_1, \dots, r_\ell)$ such that $r_i \in \{0,1\}^{f(n)}$, for each $i\in \{1, \dots, \ell\}$. For $v\in G$ and $r\in R$ and a fixed prover $\mathcal{M}$, we call $\tau(\mathcal{M}, v, r)$ the transcript $\tau(\mathcal{M}, v)$ such that $\tau_{2i-1} = r_i$ when $k$ is even and $\tau_{2i}=r_i$ otherwise, for each $i\in \{1, \dots, \ell\}$. In full words, $\tau(\mathcal{M}, v, r)$ is the transcript of the protocol, when the nodes draw the random strings from $r$. 
	
	We can construct a one-round protocol $\mathcal{P}'$, where the prover sends to each node $v$ the following certificate:
	
	\begin{enumerate}
		
		\item A spanning tree $T$ given by the $\id$ of a root $\rho$, the parent of $v$ in the tree, denoted by $t_v$, and the distance in $T$ 
		from $\rho$ to $v$, given by $d_v$.
		
		\item The list $m_v = \{m^v_{r}\}_{r\in R}$, where $m^v_r \in \{0,1\}^{k f(n)}$ is interpreted as $\tau(\mathcal{M}, v, r)$.

		\item A vector $\acc(v)\in \{0,1\}^{|R|}$ where $\acc(v)_{r}$ indicates that $u$ accept in the transcript given by $m^u_r$, for all $u$ in the subtree $T_v$ associated to $v$.

	\end{enumerate}

	Given the messages received from the prover, the nodes first verify the consistency of the tree given by \lipics{(1)}, following the spanning tree protocol given in \cite{korman2010proof}. Then, each node $v$ checks that for each $r \in R$ the given transcript $m_v^r$ is consistent with $r$. Then, for each $r\in R$, each node simulates the $k$ rounds of protocol $\mathcal{P}$ using the certificates of its  neighborhood, and decide whether to accept or reject. That information is stored in a vector $a^v \in \{0,1\}^{|R|}$. In order to check the consistency of the vector $\acc(v)$, for each $r\in R$ we say that $\acc(v)_r=1$ if and only if $a^v_r = 1$ and $\acc(u)_r=1$ for every children $u$ in $T_v$. 
	If all previous conditions are satisfied and $v$ is not the root, then $v$ accepts. Finally, the root $\rho$ verifies previous conditions and counts the number of accepting entries in $\acc(\rho)$ and accepts if they are at least two-thirds of the total. In any other case, the nodes reject.

	The number of bits sent by the prover is: $\cO(\log n)$ in \lipics{(1)}, $(k f(n))\cdot 2^{\cO(k\cdot f(n))} = 2^{\cO(k\cdot f(n))}$ in \lipics{(2)} and $2^{\cO(k\cdot f(n))}$ in \lipics{(3)}. So, in total, the number of bits communicated in any round is $2^{\cO(k\cdot f(n))}+\log n$. We now explain the completeness and soundness.
	
	\begin{itemize}
		\item \textbf{{Completeness:}} If an instance $(G,\id, I)$ is in $\mathcal{L}$, an honest prover will send the real answers that each node would have received in the $k$-round protocol, for which at least two-thirds of the coins all nodes accept, therefore the root accepts.
		
		\medskip
		
		\item \textbf{{Soundness:}} Suppose now that $(G,\id, I)$ is not in $ \mathcal{L}$, and suppose by contradiction that there exist a prover 
		$\tilde{\mathcal{M}}$ of protocol $\mathcal{P}'$ accepted by all vertices. Let $m_v$ be the certificate that $\tilde{\mathcal{M}}$ gives to vertex $v$ given by \lipics{(2)}. Now, let $\hat{\mathcal{M}}$ be a prover of $\mathcal{P}$ such that $\tau(\hat{\mathcal{M}},v,r) = m_v^r$, for each $r\in R$. Since the root accepts, all nodes must accept two thirds of the transcripts, which contradicts the soundness of $\mathcal{P}$. 
		
	\end{itemize}
\end{proof}

A direct consequence of previous result is the transfer of lower bound from non-determinism to distributed interactive protocols with shared randomness.

\begin{corollary}
	\label{lowerBound}
	Let $k \ge 1$ and let $\mathcal{L}$ be a language such that $\mathcal{L}  \in \dM[\Omega(f(n))]$, where $f(n) =\omega(\log n)$. Then,
$\mathcal{L} \in \dAM^{\pub}[k, \Omega(\frac{\log f(n)}{k})] = \dAM^{\pub}[k, \Omega(\log f(n))]$. 
\end{corollary}

\begin{corollary}
	Let $k \ge 1$. Then, problems \textsc{symmetry}, \textsc{diameter}, $\overline{3\text{-}\textsc{col}}$, $\triangle\text{-}\textsc{free}$ $\in \dAM^{\pub}[k,\Omega(\log n)]$. Also, 
	$\textsc{mst} \in \dAM^{\pub}[k,\Omega(\log \log n)]$.
\end{corollary}

\begin{proof}
We just need to apply already known lower bounds: $\textsc{symmetry} \in \dM[\Omega(n)]$ from~\cite{goos2016locally}, $\textsc{diameter} \in \dM[\Omega(n)]$ from~\cite{censor2018approximate}, $\overline{3\text{-}\textsc{col}} \in \dM[\Omega(n)]$ from~\cite{goos2016locally}, 
$\triangle\text{-}\textsc{free} \in \dM[\Omega(n)]$ from~\cite{crescenzi2019trade}, $\textsc{mst} \in  \dM[\Omega(\log^2 n)]$ from~\cite{korman2007distributed}.
\end{proof}

\begin{remark}
The lower bound saying that $\textsc{symmetry} \in \dAM^{\pub}[k,\Omega(\log n)]$ is tight. 
More precisely, the $\dMAM^{\priv}[\log n]$ protocol given by Kol, Oshman and Saxena~\cite{kol2018interactive} for solving \symmetry\ can be easily adapted to work with shared randomness. In fact, the protocol  is somehow designed in that way, where one particular node generates the random string and shares it with the other nodes (through Merlin).
 Therefore, $\symmetry \in \dMAM^{\pub}[\log n]$. On the other hand,  $\symmetry \in \dM[\Omega(n^2)]$~\cite{goos2016locally}.
\end{remark}

In the proof of Theorem~\ref{thm:simula}, in order to design a \dM\ protocol, we had to construct  a spanning tree for verifying
 that two thirds of all coins are accepted by {\emph{all nodes}}. In fact,  it could be the case that, for negative instances,
every  node rejects a very small portion of the coins, getting the wrong idea that the instance is positive. For avoiding that, 
and coordinating the nodes,  in the $\dM$ protocol we construct a spanning tree. This is where the additive $\log n$ term comes from. 
Next result states that previous situation does not
occur if, instead of two thirds,  we ask the interactive protocol to accept \textit{with high probabilty}.

\begin{theorem}\label{thm:whp}
Let $k\geq 1$ and let $\mathcal{L}$ be a language such that  
$\mathcal{L} \in \dAM^{\pub}_{\eps}[k, f(n)]$, with $\eps < \frac{1}{n+1}$. Then, $\mathcal{L} \in \dM[2^{\cO(k\; f(n))}]$. 
\end{theorem}

\begin{proof}
	Let $\mathcal{L}$ be a language over instance $(G, \id, I)$ with $n(G)=n$ and let $\mathcal{P}$ be a $k$ round protocol for $\mathcal{L}$ with $f(n)$ bits such that its acceptante error is less than $\frac{1}{m}$, where $m> n+1$.
	
	In order to construct the protocol, we proceed in a similar way as the previous theorem: Merlin sends each node an enumeration $m_{r}^v$ of the answers to each possible coin that the original prover sends each node. With this all nodes share their ceritificates to each coin, simulate all of them and each node $v$ accepts iff $(1-\frac{1}{m})$  of the coins are accepted by him.
	
	We have that if $G$ is a  \emph{yes instance}, an honest prover will return the answers to the original protocol and all nodes will accept $(1-\frac{1}{m})$ of all possible coins, therefore all accept the protocol.
	
	If $G$ is a \textit{no instance} and all nodes accept the protocol, we have that all nodes reject at most $\frac{2^{k\:f(n)}}{m}$ coins, therefore the total amount of coins that are rejected by some node is at most $\frac{n2^{k\: f(n)}}{m}$, which is strictly less than $(1-\frac{1}{m})2^{k\:f(n)}$.
\end{proof}

\begin{corollary}\label{cor:lowerwhp}
Let $k \ge 1$ and let $\mathcal{L}$ be a language such that $\mathcal{L}  \in \dM[\Omega(f(n))]$. Then, 
$\mathcal{L} \in \dAM^{\pub}_{\eps}[k, \Omega(\frac{\log f(n)}{k})] = \dAM^{\pub}_{\eps}[k, \Omega(\log f(n))]$,  with $\eps < \frac{1}{n+1}$.
\end{corollary} 

\begin{corollary}
	Let $k \ge 1$. Then,  problems  \textsc{planar}, \textsc{outerplanar}, \textsc{spanning-tree} $\in \dAM^{\pub}_{\eps}[k,\Omega(\log \log n)]$,
	with $\eps < \frac{1}{n+1}$.
\end{corollary}

\begin{proof}
All these languages belong to $\dM[\log n]$~\cite{goos2016locally}.
\end{proof}


\section{ \texorpdfstring{dAM$^{\pub}$ v\/s dAM$^{\priv}$}{shared dAM vs private dAM}}\label{sec:am}
A recent result  shows that \dAM\ protocols with private randomness are more powerful than \dAM\ protocols with shared randomness~\cite{crescenzi2019trade}.
The precise result corresponds to next proposition.

\begin{proposition}[\cite{crescenzi2019trade}]\label{prop:cresc}
Let $k \ge 1$, $0< \eps < \frac{1}{2}$, and $\mathcal{L}$ be a language such that  $\mathcal{L} \in \dAM_{\eps}^{\pub}[k, f(n)]$. Then,   $\mathcal{L} \in \dAM_{\eps}^{\priv}[k, f(n) + \log n]$.
\end{proposition}

A natural question is whether the two models are equivalent. In this section we give a negative answer. We separate them through problem \amos. Recall that
\amos\ is the language of labeled graphs where at most  one node is selected. It is already known that $\amos \in \dM[\Theta(\log n)]$~\cite{goos2016locally}. Moreover,
in~\cite{fraigniaud2019distributed} the authors show that adding randomness {\emph{after}} the nondeterministic round does not help. 
More precisely, $\amos \in \dMA^{\priv}_{\eps} [\Omega(\log n)]$, for  $0< \eps < \frac{1}{5}$.

The situation changes dramatically when randomness goes {\emph{before}} nondeterminism, as explained in the following proposition.

\begin{proposition}[\cite{fraigniaud2019distributed}]
Let $0< \eps < \frac{1}{2}$. Then,  $\amos \in \dAM_{\eps}^{\priv}[\log(\eps^{-1})]=\dAM_{\eps}^{\priv}[1]$.
\end{proposition}

In the shared randomness framework, we can construct a protocol that uses bandwidth $O(\log\log n)$. As we are going to see 
in Theorem~\ref{thm:amos}, this upper bound is indeed tight.

\begin{lemma}\label{lem:AMOSloglogupper}
$\amos \in\dAM^{\pub}[\log \log n]$. 
\end{lemma}

\begin{proof}
The protocol is the following. First, each node considers the smallest prime $q$ such that $\log^{c+2}n \leq q \leq2\log^{c+2}n$  and constructs a polynomial 
over the field $\mathbb{F}_q$ associated to its  \id\ given by $p_v(x)=\sum_{i\leq\log (\id(v))}\textsf{bin}_i(\id(v)) \cdot x^i$.  Where $\textsf{bin}_i(m)$ corresponds to i-th bit in the binary representation of $m$.
All nodes generate a random string $s \in \mathbb{F}_q$ using the shared randomness.
Then, the prover sends to each node the random evaluation  of the selected node $v_0$. More precisely, $\bar{p}= p_{v_0}(s)$,
which is of size $O(\log \log n)$. The nodes first check they received the same value of $\bar{p}$. If a node $v$ is not selected, then it always accepts; otherwise, it accepts if and only if $p_v(s) = p_{v_0}(s)$.
If an instance belongs to  \textsc{amos}, then all nodes accept. Otherwise, there will exist at least two selected nodes that accept, such that their \id's polinomial matches on $s$ with probability at most $\frac{1}{\log^cn}$.
\end{proof}

From Corollary \ref{cor:lowerwhp} we conclude that, for every $k \ge1$, $\amos \in \dAM^{\pub}_{\eps}[k,\Omega(\log \log n)]$ with $\eps < \frac{1}{n+1}$.
In other words, the protocol given in Lemma \ref{lem:AMOSloglogupper} matches the lower bound for all correct protocols that run with high  probability.
Next theorem shows that the upper bound is matched even when $\eps=\frac{1}{3}$. 
 
\begin{theorem}\label{thm:amos}
	Let $k \ge 1$. Then,  $\textsc{amos} \in \dAM^{\pub}[k, \Theta(\log \log n)]$. 
\end{theorem}

\begin{proof}
	We follow a tecnique by \cite{goos2016locally} for "glueing" solutions together to form a bad instance, with some modifications.
	Without loss of generality, we may assume that $n$ is even. Let $A$ be a partition of $[1,n^2]$ into $n$ sets of size $n$, and let $B$ be a partition of $[n^2+1 , 2n^2]$ in a similar manner. Let $\mathcal{G}$ be a family of $n$-node graphs.
	
	Set now $\mathcal{G}_A $ to be the set of labeled graphs in $\mathcal{G}$, with labeled picked from $A$. Let $F_a$ be a graph in $G_A$, and let $v^*$ be the vertex of $F_a$ labelled with the smallest label. Consider the input $I_a$ such that $I_a(v) = 0$ for every vertex $v\in V(F_a)$ except for the vertex $v^*$, for which $I_a(v^*) = 1$. Similarly, we define $\mathcal{G}_B$ as the set of graphs in $\mathcal{G}$ labelled with labels in $B$. For each graph $F_b \in  \mathcal{G}_B$, we define the input $I_b$ such that $I_b(v) = 0$ for all $v\in V(F_b)$.

	For $(F_a,F_b) \in \mathcal{G}_A\times \mathcal{G}_B$ let $G(F_a,F_b)$ be the graph defined by the disjoint union of graphs  $F_a$ and $F_b$ plus four additional nodes $x_A, y_A, x_B, y_B$. These nodes are labeled with different numbers in the set $C = [2n^2 +1, 3n^2]$. Nodes $x_A$ and $x_B$ are both adjacent only to $y_A$ and $y_B$. Node $y_A$ is adjacent only to $x_A, x_B$ and $v_a$, where $v_a$ is some node of $F_a$. Node $y_B$ is adjacent only to $x_A, x_B$ and $v_b$, where $v_b$ is some node of $F_b$.  Observe  that $v^* \neq v_a$ and all nodes in $F_a$ communicate with $F_b$ only through the nodes $x_A, y_A, x_B, y_B$.

	Consider now the input $I$ of $G(F_a, F_b)$ such that $I(v) = I_a(v)$ if $v\in V(F_a)$, $I(v) = 0$ otherwise. Observe that  $v^* \in V(F_a)$ is the only node in $V(G(F_a, F_b))$ satisfying $I(v^*)= 1$. Therefore, the graph $G=G(F_a,F_b) $ is a Yes-instance of \textsc{amos}.
	
	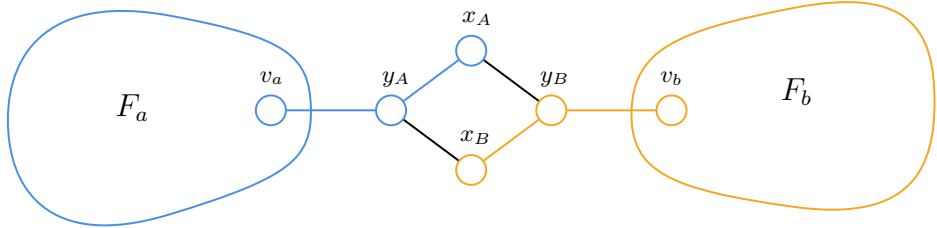
\begin{figure}[!ht]
		\centering

\tikzset{every picture/.style={line width=0.75pt}} 

\begin{tikzpicture}[x=0.75pt,y=0.75pt,yscale=-1,xscale=1]

\draw  [color={rgb, 255:red, 74; green, 144; blue, 226 }  ,draw opacity=1 ] (109.72,102.41) .. controls (168.72,113.41) and (179.22,122.91) .. (179.97,148.66) .. controls (180.72,174.41) and (169.72,183.91) .. (109.22,202.41) .. controls (48.72,220.91) and (28.47,189.16) .. (28.72,154.41) .. controls (28.97,119.66) and (50.72,91.41) .. (109.72,102.41) -- cycle ;
\draw  [color={rgb, 255:red, 245; green, 166; blue, 35 }  ,draw opacity=1 ] (421.03,98.48) .. controls (480.03,86.98) and (491.03,110.98) .. (491.28,144.73) .. controls (491.53,178.48) and (480.53,206.98) .. (420.53,198.48) .. controls (360.53,189.98) and (339.03,181.48) .. (340.03,150.48) .. controls (341.03,119.48) and (362.03,109.98) .. (421.03,98.48) -- cycle ;
\draw [color={rgb, 255:red, 74; green, 144; blue, 226 }  ,draw opacity=1 ]   (160,150) -- (220,150) ;
\draw [color={rgb, 255:red, 74; green, 144; blue, 226 }  ,draw opacity=1 ]   (260,120) -- (220,150) ;
\draw    (260,120) -- (300,150) ;
\draw [color={rgb, 255:red, 245; green, 166; blue, 35 }  ,draw opacity=1 ]   (300,150) -- (260,180) ;
\draw    (260,180) -- (220,150) ;
\draw [color={rgb, 255:red, 245; green, 166; blue, 35 }  ,draw opacity=1 ]   (360,150) -- (300,150) ;
\draw  [color={rgb, 255:red, 74; green, 144; blue, 226 }  ,draw opacity=1 ][fill={rgb, 255:red, 255; green, 255; blue, 255 }  ,fill opacity=1 ] (212.5,150) .. controls (212.5,145.86) and (215.86,142.5) .. (220,142.5) .. controls (224.14,142.5) and (227.5,145.86) .. (227.5,150) .. controls (227.5,154.14) and (224.14,157.5) .. (220,157.5) .. controls (215.86,157.5) and (212.5,154.14) .. (212.5,150) -- cycle ;
\draw  [color={rgb, 255:red, 74; green, 144; blue, 226 }  ,draw opacity=1 ][fill={rgb, 255:red, 255; green, 255; blue, 255 }  ,fill opacity=1 ] (152.5,150) .. controls (152.5,145.86) and (155.86,142.5) .. (160,142.5) .. controls (164.14,142.5) and (167.5,145.86) .. (167.5,150) .. controls (167.5,154.14) and (164.14,157.5) .. (160,157.5) .. controls (155.86,157.5) and (152.5,154.14) .. (152.5,150) -- cycle ;
\draw  [color={rgb, 255:red, 74; green, 144; blue, 226 }  ,draw opacity=1 ][fill={rgb, 255:red, 255; green, 255; blue, 255 }  ,fill opacity=1 ] (252.5,120) .. controls (252.5,115.86) and (255.86,112.5) .. (260,112.5) .. controls (264.14,112.5) and (267.5,115.86) .. (267.5,120) .. controls (267.5,124.14) and (264.14,127.5) .. (260,127.5) .. controls (255.86,127.5) and (252.5,124.14) .. (252.5,120) -- cycle ;
\draw  [color={rgb, 255:red, 245; green, 166; blue, 35 }  ,draw opacity=1 ][fill={rgb, 255:red, 255; green, 255; blue, 255 }  ,fill opacity=1 ] (252.5,180) .. controls (252.5,175.86) and (255.86,172.5) .. (260,172.5) .. controls (264.14,172.5) and (267.5,175.86) .. (267.5,180) .. controls (267.5,184.14) and (264.14,187.5) .. (260,187.5) .. controls (255.86,187.5) and (252.5,184.14) .. (252.5,180) -- cycle ;
\draw  [color={rgb, 255:red, 245; green, 166; blue, 35 }  ,draw opacity=1 ][fill={rgb, 255:red, 255; green, 255; blue, 255 }  ,fill opacity=1 ] (292.5,150) .. controls (292.5,145.86) and (295.86,142.5) .. (300,142.5) .. controls (304.14,142.5) and (307.5,145.86) .. (307.5,150) .. controls (307.5,154.14) and (304.14,157.5) .. (300,157.5) .. controls (295.86,157.5) and (292.5,154.14) .. (292.5,150) -- cycle ;
\draw  [color={rgb, 255:red, 245; green, 166; blue, 35 }  ,draw opacity=1 ][fill={rgb, 255:red, 255; green, 255; blue, 255 }  ,fill opacity=1 ] (352.5,150) .. controls (352.5,145.86) and (355.86,142.5) .. (360,142.5) .. controls (364.14,142.5) and (367.5,145.86) .. (367.5,150) .. controls (367.5,154.14) and (364.14,157.5) .. (360,157.5) .. controls (355.86,157.5) and (352.5,154.14) .. (352.5,150) -- cycle ;

\draw (81,140.4) node [anchor=north west][inner sep=0.75pt]  [font=\Large]  {$F_{a}$};
\draw (413,132.4) node [anchor=north west][inner sep=0.75pt]  [font=\Large]  {$F_{b}$};
\draw (153,129) node [anchor=north west][inner sep=0.75pt]  [font=\small]  {$v_{a}$};
\draw (214,129) node [anchor=north west][inner sep=0.75pt]  [font=\small]  {$y_{A}$};
\draw (254,99) node [anchor=north west][inner sep=0.75pt]  [font=\small]  {$x_{A}$};
\draw (293,129) node [anchor=north west][inner sep=0.75pt]  [font=\small]  {$y_{B}$};
\draw (253,159) node [anchor=north west][inner sep=0.75pt]  [font=\small]  {$x_{B}$};
\draw (353,129) node [anchor=north west][inner sep=0.75pt]  [font=\small]  {$v_{b}$};

\end{tikzpicture}

		\caption{The auxiliary graph $G(F_a,F_b)$, with $a\in A$ and $b\in B$. This is a yes-instance for \textsc{amos}, as there is only one selected node (in $F_a$).}
		\label{fig:amos}
	\end{figure}
 
 Let $\mathcal{P}$ be a $k$-round distributed interactive proof with shared randomness verifying $\textsc{amos}$ with bandwidth $K = \delta \log \log n$ and error probability $\eps$. 
 Let us call $\{x_A,y_A,x_B,y_B\}$ the \emph{bridge} of $G(F_a, F_b)$. We can assume, without loss of generality, that $\mathcal{P}$ satisfies that for any (shared) random string generated by Arthur, the nodes in the bridge $\{x_A,y_A,x_B,y_B\}$ receive the same proof. Indeed, if we have a protocol that is not simple, with cost  $L$, we can design a new protocol that is simple and whose proof has length $4L$ by making each node pick their portion of the proof and going by the original protocol afterwards.
	
Given sequence of random strings $r=(r_1,r_2,\dots r_k)$, we call $m^r$ the sequence indexed by vertices $v\in V(G[F_a,F_b])$, such that $m^r_v$ is the set of certificates that Merlin sends to node $v$ in protocol $\mathcal{P}$, when Arthur communicate string $r_i$ on round~$i$. 
	Let $m_{ab}: \{0,1\}^{Kk} \to \{0,1\}^{Kk}$ be the function that associates to each sequence $r=(r_1,r_2,\dots r_k)$ the tuple $(m^r_{x_A}, m^r_{x_B}, m^r_{y_A}, m^r_{y_B})$ such that it extends to a proof assignment for the nodes in both $F_{a}$ and $F_{b}.$ that make them accept whenever the bridge accepts.
	
	
	
	Now consider the complete bipartite graph $\hat{G} = A \cup B$. For each $a\in A$ and  $b\in B$, color the edge $\{a,b\}$ with function $m_{ab}$. There are at most $2^{Kk2^{Kk}}$ possible functions. Therefore, by the pigeonhole principle, there exists a monochromatic set of edges $W$ of size at least $\frac{n^2}{2^{Kk2^{Kk}}}$. 
	
	Observe that  for sufficiently small $\delta$ and large $n$, $2^{Kk2^{Kk}} = (\log n)^{\delta k \log^{\delta k} (n)} = o(n^{1/2})$. Indeed, if $n > 2^{k\delta}$ and $\delta < 1/(4k)$  have that $ \delta k\log^{\delta k}(n) \log \log (n )\leq \log^{2\delta k}(n) <\frac{1}{2} \log n$.  Following  a result of Bondy and Simonovits  given in \cite{bondy1974cycles}, we have that there exists a $4$-cycle $a_1,b_1, a_2, b_2$  in the subgraph $\hat{G}$  induced by $W$. 
	
	Consider now the graph $G(a_1, b_1, a_2, b_2)$ defined as follows: First. take a disjoint union of $F_{a_1}, F_{b_1}, F_{a_2}$ and $F_{b_2}$. Then, for each $i \in \{1, 2\}$ add nodes $x_A^{i}, x_B^{i}, y_A^{i}, y_B^{i}$, labelled with different labels in $[2n^2+1, 3n^2]$  correspondant to the yes instances formed by $F_{a_i}$ and $F_{b_j}$. 	For each $i\in \{1,2\}$, the node $y_A^{i}$  is adjacent to $x_A^i$, $x_B^{i+1}$ and the node $v_{a_i}$ of $F_{a_i}$. Similarly $y_B^i$ is adjacent to the $x_A^i$, $x_B^{i+ 1}$ and the node $v_{b_i}$ of $F_{b_i}$. Where the $i+1$ is taken$\mod 2$. Finally, define the input $I$ as $I(v) = I_{a_i}(v)$ if $v$ belongs to $F_{a_i}$ and $I(v) = 0$ otherwise. 
	
	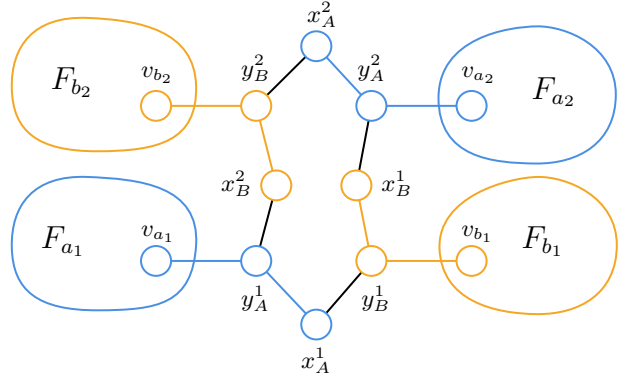
\begin{figure}[!ht]
		\centering

\tikzset{every picture/.style={line width=0.75pt}} 

\begin{tikzpicture}[x=0.75pt,y=0.75pt,yscale=-1,xscale=1]

\draw [color={rgb, 255:red, 245; green, 166; blue, 35 }  ,draw opacity=1 ]   (160,190) -- (210,190) ;
\draw  [color={rgb, 255:red, 245; green, 166; blue, 35 }  ,draw opacity=1 ][fill={rgb, 255:red, 255; green, 255; blue, 255 }  ,fill opacity=1 ] (152.5,190) .. controls (152.5,185.86) and (155.86,182.5) .. (160,182.5) .. controls (164.14,182.5) and (167.5,185.86) .. (167.5,190) .. controls (167.5,194.14) and (164.14,197.5) .. (160,197.5) .. controls (155.86,197.5) and (152.5,194.14) .. (152.5,190) -- cycle ;
\draw  [color={rgb, 255:red, 245; green, 166; blue, 35 }  ,draw opacity=1 ] (125.5,145.75) .. controls (152,146.75) and (179.33,152.08) .. (180,180) .. controls (180.67,207.92) and (161,212.75) .. (130.83,212.83) .. controls (100.67,212.92) and (88.5,202.75) .. (88,181.25) .. controls (87.5,159.75) and (99,144.75) .. (125.5,145.75) -- cycle ;
\draw  [color={rgb, 255:red, 74; green, 144; blue, 226 }  ,draw opacity=1 ] (350,150) .. controls (364.5,149.75) and (389.5,160) .. (389.5,184.25) .. controls (389.5,208.5) and (371,219.25) .. (349.5,219) .. controls (328,218.75) and (301,216.25) .. (301,186) .. controls (301,155.75) and (335.5,150.25) .. (350,150) -- cycle ;
\draw [color={rgb, 255:red, 74; green, 144; blue, 226 }  ,draw opacity=1 ]   (317.5,190) -- (267.5,190) ;
\draw  [color={rgb, 255:red, 74; green, 144; blue, 226 }  ,draw opacity=1 ][fill={rgb, 255:red, 255; green, 255; blue, 255 }  ,fill opacity=1 ] (310,190) .. controls (310,185.86) and (313.36,182.5) .. (317.5,182.5) .. controls (321.64,182.5) and (325,185.86) .. (325,190) .. controls (325,194.14) and (321.64,197.5) .. (317.5,197.5) .. controls (313.36,197.5) and (310,194.14) .. (310,190) -- cycle ;
\draw [color={rgb, 255:red, 0; green, 0; blue, 0 }  ,draw opacity=1 ]   (210,190) -- (240,160) ;
\draw [color={rgb, 255:red, 74; green, 144; blue, 226 }  ,draw opacity=1 ]   (240,160) -- (267.5,190) ;
\draw  [color={rgb, 255:red, 74; green, 144; blue, 226 }  ,draw opacity=1 ][fill={rgb, 255:red, 255; green, 255; blue, 255 }  ,fill opacity=1 ] (232.5,160) .. controls (232.5,155.86) and (235.86,152.5) .. (240,152.5) .. controls (244.14,152.5) and (247.5,155.86) .. (247.5,160) .. controls (247.5,164.14) and (244.14,167.5) .. (240,167.5) .. controls (235.86,167.5) and (232.5,164.14) .. (232.5,160) -- cycle ;
\draw [color={rgb, 255:red, 0; green, 0; blue, 0 }  ,draw opacity=1 ]   (210,268) -- (220,230) ;
\draw [color={rgb, 255:red, 245; green, 166; blue, 35 }  ,draw opacity=1 ]   (220,230) -- (210,190) ;
\draw  [color={rgb, 255:red, 245; green, 166; blue, 35 }  ,draw opacity=1 ][fill={rgb, 255:red, 255; green, 255; blue, 255 }  ,fill opacity=1 ] (212.5,230) .. controls (212.5,225.86) and (215.86,222.5) .. (220,222.5) .. controls (224.14,222.5) and (227.5,225.86) .. (227.5,230) .. controls (227.5,234.14) and (224.14,237.5) .. (220,237.5) .. controls (215.86,237.5) and (212.5,234.14) .. (212.5,230) -- cycle ;
\draw [color={rgb, 255:red, 74; green, 144; blue, 226 }  ,draw opacity=1 ]   (160,268) -- (210,268) ;
\draw  [color={rgb, 255:red, 74; green, 144; blue, 226 }  ,draw opacity=1 ][fill={rgb, 255:red, 255; green, 255; blue, 255 }  ,fill opacity=1 ] (152.5,268) .. controls (152.5,263.86) and (155.86,260.5) .. (160,260.5) .. controls (164.14,260.5) and (167.5,263.86) .. (167.5,268) .. controls (167.5,272.14) and (164.14,275.5) .. (160,275.5) .. controls (155.86,275.5) and (152.5,272.14) .. (152.5,268) -- cycle ;
\draw  [color={rgb, 255:red, 74; green, 144; blue, 226 }  ,draw opacity=1 ] (125.5,225.75) .. controls (152,226.75) and (181.83,230.08) .. (180,260) .. controls (178.17,289.92) and (161,292.75) .. (130.83,292.83) .. controls (100.67,292.92) and (88.5,282.75) .. (88,261.25) .. controls (87.5,239.75) and (99,224.75) .. (125.5,225.75) -- cycle ;
\draw  [color={rgb, 255:red, 245; green, 166; blue, 35 }  ,draw opacity=1 ] (350.5,225.75) .. controls (365,225.5) and (390,235.75) .. (390,260) .. controls (390,284.25) and (371.5,295) .. (350,294.75) .. controls (328.5,294.5) and (302.5,288) .. (301.5,261.75) .. controls (300.5,235.5) and (336,226) .. (350.5,225.75) -- cycle ;
\draw [color={rgb, 255:red, 245; green, 166; blue, 35 }  ,draw opacity=1 ]   (317.5,268) -- (267.5,268) ;
\draw  [color={rgb, 255:red, 245; green, 166; blue, 35 }  ,draw opacity=1 ][fill={rgb, 255:red, 255; green, 255; blue, 255 }  ,fill opacity=1 ] (310,268) .. controls (310,263.86) and (313.36,260.5) .. (317.5,260.5) .. controls (321.64,260.5) and (325,263.86) .. (325,268) .. controls (325,272.14) and (321.64,275.5) .. (317.5,275.5) .. controls (313.36,275.5) and (310,272.14) .. (310,268) -- cycle ;
\draw [color={rgb, 255:red, 74; green, 144; blue, 226 }  ,draw opacity=1 ]   (210,268) -- (240,300) ;
\draw [color={rgb, 255:red, 0; green, 0; blue, 0 }  ,draw opacity=1 ]   (240,300) -- (267.5,268) ;
\draw  [color={rgb, 255:red, 74; green, 144; blue, 226 }  ,draw opacity=1 ][fill={rgb, 255:red, 255; green, 255; blue, 255 }  ,fill opacity=1 ] (232.5,300) .. controls (232.5,295.86) and (235.86,292.5) .. (240,292.5) .. controls (244.14,292.5) and (247.5,295.86) .. (247.5,300) .. controls (247.5,304.14) and (244.14,307.5) .. (240,307.5) .. controls (235.86,307.5) and (232.5,304.14) .. (232.5,300) -- cycle ;
\draw  [color={rgb, 255:red, 74; green, 144; blue, 226 }  ,draw opacity=1 ][fill={rgb, 255:red, 255; green, 255; blue, 255 }  ,fill opacity=1 ] (202.5,268) .. controls (202.5,263.86) and (205.86,260.5) .. (210,260.5) .. controls (214.14,260.5) and (217.5,263.86) .. (217.5,268) .. controls (217.5,272.14) and (214.14,275.5) .. (210,275.5) .. controls (205.86,275.5) and (202.5,272.14) .. (202.5,268) -- cycle ;
\draw [color={rgb, 255:red, 0; green, 0; blue, 0 }  ,draw opacity=1 ]   (260,230) -- (267.5,190) ;
\draw [color={rgb, 255:red, 245; green, 166; blue, 35 }  ,draw opacity=1 ]   (267.5,268) -- (260,230) ;
\draw  [color={rgb, 255:red, 245; green, 166; blue, 35 }  ,draw opacity=1 ][fill={rgb, 255:red, 255; green, 255; blue, 255 }  ,fill opacity=1 ] (252.5,230) .. controls (252.5,225.86) and (255.86,222.5) .. (260,222.5) .. controls (264.14,222.5) and (267.5,225.86) .. (267.5,230) .. controls (267.5,234.14) and (264.14,237.5) .. (260,237.5) .. controls (255.86,237.5) and (252.5,234.14) .. (252.5,230) -- cycle ;
\draw  [color={rgb, 255:red, 245; green, 166; blue, 35 }  ,draw opacity=1 ][fill={rgb, 255:red, 255; green, 255; blue, 255 }  ,fill opacity=1 ] (260,268) .. controls (260,263.86) and (263.36,260.5) .. (267.5,260.5) .. controls (271.64,260.5) and (275,263.86) .. (275,268) .. controls (275,272.14) and (271.64,275.5) .. (267.5,275.5) .. controls (263.36,275.5) and (260,272.14) .. (260,268) -- cycle ;
\draw  [color={rgb, 255:red, 74; green, 144; blue, 226 }  ,draw opacity=1 ][fill={rgb, 255:red, 255; green, 255; blue, 255 }  ,fill opacity=1 ] (260,190) .. controls (260,185.86) and (263.36,182.5) .. (267.5,182.5) .. controls (271.64,182.5) and (275,185.86) .. (275,190) .. controls (275,194.14) and (271.64,197.5) .. (267.5,197.5) .. controls (263.36,197.5) and (260,194.14) .. (260,190) -- cycle ;
\draw  [color={rgb, 255:red, 245; green, 166; blue, 35 }  ,draw opacity=1 ][fill={rgb, 255:red, 255; green, 255; blue, 255 }  ,fill opacity=1 ] (202.5,190) .. controls (202.5,185.86) and (205.86,182.5) .. (210,182.5) .. controls (214.14,182.5) and (217.5,185.86) .. (217.5,190) .. controls (217.5,194.14) and (214.14,197.5) .. (210,197.5) .. controls (205.86,197.5) and (202.5,194.14) .. (202.5,190) -- cycle ;

\draw (234,136) node [anchor=north west][inner sep=0.75pt]  [font=\small]  {$x_{A}^{2}$};
\draw (151,249) node [anchor=north west][inner sep=0.75pt]  [font=\small]  {$v_{a_{1}}$};
\draw (106,169) node [anchor=north west][inner sep=0.75pt]  [font=\Large]  {$F_{b_{2}}$};
\draw (346,174) node [anchor=north west][inner sep=0.75pt]  [font=\Large]  {$F_{a_{2}}$};
\draw (101,249) node [anchor=north west][inner sep=0.75pt]  [font=\Large]  {$F_{a_{1}}$};
\draw (341,249) node [anchor=north west][inner sep=0.75pt]  [font=\Large]  {$F_{b_{1}}$};
\draw (259,163) node [anchor=north west][inner sep=0.75pt]  [font=\small]  {$y_{A}^{2}$};
\draw (231,310) node [anchor=north west][inner sep=0.75pt]  [font=\small]  {$x_{A}^{1}$};
\draw (201,279) node [anchor=north west][inner sep=0.75pt]  [font=\small]  {$y_{A}^{1}$};
\draw (311,169) node [anchor=north west][inner sep=0.75pt]  [font=\small]  {$v_{a_{2}}$};
\draw (201,163) node [anchor=north west][inner sep=0.75pt]  [font=\small]  {$y_{B}^{2}$};
\draw (261,279) node [anchor=north west][inner sep=0.75pt]  [font=\small]  {$y_{B}^{1}$};
\draw (191,221) node [anchor=north west][inner sep=0.75pt]  [font=\small]  {$x_{B}^2$};
\draw (271,221) node [anchor=north west][inner sep=0.75pt]  [font=\small]  {$x_{B}^1$};
\draw (311,249) node [anchor=north west][inner sep=0.75pt]  [font=\small]  {$v_{b_{1}}$};
\draw (151,169) node [anchor=north west][inner sep=0.75pt]  [font=\small]  {$v_{b_{2}}$};

\end{tikzpicture}

		\caption{A No-instance for \textsc{amos}, with messages of size $o(\log n)$. From every node's perspective, the graph behaves like a Yes-instance, but there are two selected nodes.}
		\label{fig:fakeAmos}
	\end{figure}
	
	
	Observe that $G(a_1, b_1, a_2, b_2)$ is a No-instance of \textsc{amos}, as there are two selected vertices: one in $V(F_{a_1})$ and another one in $V(F_{a_2})$. Moreover,  all nodes of the set $x_A^{1}, x_B^{1}, y_A^{1}, y_B^{1}$, $x_A^{2}, x_B^{2}, y_A^{2}, y_B^{2}$ receive the same answers by Merlin which extends to assignments for the nodes in $F_{a_i}$ and $F_{b_i}$ that make them accept with the same probability as the nodes in the bridge, as they locally place themselves in a previously defined yes instance. Therefore all vertices accept two thirds of all possible random coins. This contradicts the fact that $\mathcal{P}$ was a correct distributed interactive proof for $\textsc{amos}$.
\end{proof} 

 
\section{\texorpdfstring{dMA$^{\pub}$ v\/s dMA$^{\priv}$}{shared dMA vs private dMA}}\label{sec:ma}
In this section we first show that, in what regards private and shared randomness, roles are reversed when we address \dMA\ protocols instead of \dAM\ protocols.
In fact, we get a result analogous to that of Crescenzi, Fraigniaud, and Paz~\cite{crescenzi2019trade} (Proposition~\ref{prop:cresc})
by proving   that \dMA\ protocols with shared randomness are more powerful than \dMA\ protocols with private randomness.

\begin{theorem}\label{dMAprivToPub}
{Let  $\eps, \delta > 0$ with $\eps+\delta < \frac{1}{2}$} and consider ${\mathcal L}$ to be a language over $n$-node graphs such that ${\mathcal L} \in \dMA^{\priv}_\eps[f(n)]$. Then,  ${\mathcal L} \in  \dAM^{\pub}_{\eps + \delta}[f(n) + \log n + \log (\delta^{-1})]$.
\end{theorem}

\begin{proof}
	Let $\mathcal{L}$ be a distributed language in $\dMA^{\priv}_\eps[f(n)]$ over a network configuration $\mathcal{I} = (G, \id, I)$, and let $\mathcal{P}$ the protocol that witnesses that membership. 
	
	Let $Z( \mathcal{I}, m, r)$ be a random variable that equals to $1$ if and only if \textit{Arthur} is wrong about the membership of $\mathcal{I}$ in $\mathcal{L}$ given the proof $m$ and the coin $r$ in an occurrence of the protocol $\mathcal{P}$. Observe that both $r$ and $m$ are a sequence of $n\cdot f(n)$ bits, with the $i$-th portion of the sequence containing the message sent or received by the node identified as the $i$-th node of $G$.

	We show by the probabilistic method that there exists a collection $\{r_i\}_{i=1}^t$ of random strings such that, for all network configuration $\mathcal{I}$, we can design a correct protocol that only relies on these coins, with a small increase in error.
	
	
	
	Indeed, let $\{r_i\}_{i=1}^t$ be a collection of random strings of length $f(n)$ and consider a network configuration $\mathcal{I}$ such that $\mathcal{I} \in \mathcal{L}$ we define the event
	\[Y_\mathcal{I}  = \{ \forall m\in  \{0,1\}^{n\cdot f(n) }\quad\mathbf{E}_i ( Z(\mathcal{I}, m , r_i)) > \varepsilon + \delta\}\] 
	where $\mathbf{E}_i (\cdot)$ is the expected value over the collection of random coins mentioned. We also define, for $\mathcal{I}\notin \mathcal{L}$, the event
	\[ N_\mathcal{I}  = \{ \exists m \in \{0,1\}^{n\cdot f(n)} \:\text{s.t. } \:\mathbf{E}_i ( Z(\mathcal{I}, m , r_i)) > \varepsilon + \delta\}\]
	Now, by the correctness of $\mathcal{P}$, we have that for ${\mathcal{I}}\in \mathcal{L}$ there exists a proof $m_{\mathcal{I}}$ such that {Arthur} errs with small probability, therefore we have that $\mathbf{E}_r (Z(\mathcal{I}, m_{\mathcal{I}}, r)) \leq \eps$ and by a Chernoff bound:
	\[ \mathbf{Pr} (Y_\mathcal{I}  ) \leq  \mathbf{Pr} \left[\left(\frac{1}{t} \sum_{i=1}^{t}  Z(G, m_{\mathcal{I}}, r_i) - \eps \right) > \delta \right]\leq 2e^{-2\delta^2 t} \]
	
	Now consider the case when ${\mathcal{I}}\notin \mathcal{L}$, for the correctness of $\mathcal{P}$ with obtain that $\mathbf{E}_r (Z(\mathcal{I}, m_{\mathcal{I}}, r)) \leq \eps$ for any $m \in \{0,1\}^{n\cdot f(n)}$. Thus, by an union bound and another Chernoff bound, we obtain:
	\[ \textbf{Pr} (N_\mathcal{I} ) = \textbf{Pr} \left[ \exists m \in \{0,1\}^{n \cdot f(n)}\; \text{s.t. } \left(\frac{1}{t}  \sum_{i=1}^{t}  Z(\mathcal{I}, m, r_i) - \eps \right) > \delta \right]\leq \ 2^{n \cdot f(n) }\cdot 2e^{-2\delta^2t }  \] 
	Therefore, if we consider $B = \bigcup_{\mathcal{I} \in \mathcal{L}} Y_\mathcal{I} \cup \bigcup_{\mathcal{I} \notin \mathcal{L}}N_\mathcal{I} $ to be the union of $Y_\mathcal{I} $ over all $\mathcal{I} \in \mathcal{L}$ and of $N_\mathcal{I} $ over all $\mathcal{I} \notin \mathcal{L}$  and make a union bound, we get:
	
	\[ \textbf{Pr}(B) \leq  \sum_{G\in \mathcal{L}}  2e^{-2\delta^2 t} + \sum_{G\notin \mathcal{L}} 2e^{-2\delta^2t} \cdot 2^{n \cdot f(n)}  \]
	Now, remember from the definition of a distributed language that we assume that the set of all possible network configurations defined over graphs of size $n$ is at most $2^{\poly(n)}$. Therefore, if we take $t = \Theta ( \frac{\poly(n) + n\cdot g(n)}{\delta^2} ) $ the probability of $B$ is strictly smaller than one.
	
	And so, there exists a collection $\{r_i\}_{i  = 1}^t$  such that for any network configuration $\mathcal{I}$ defined over graphs of $n$ nodes:
\begin{align*}
	&\mathcal{I} \in \mathcal{L} &\longrightarrow  \exists m \; \textbf{E}_i \left( Z(\mathcal{I},m , r_i)) \leq (\eps +\delta \right) \rightarrow \; \exists m \; \mathbf{Pr}(Z(\mathcal{I},m , r_i) = 0) > 1 - (\eps + \delta)\\
	&\mathcal{I} \notin \mathcal{L} &\longrightarrow \forall m \; \textbf{E}_i \left( Z(\mathcal{I},m , r_i)) \leq (\eps+ \delta \right) \rightarrow \; \forall m \; \mathbf{Pr}(Z(\mathcal{I},m , r_i) = 0) > 1 - (\eps + \delta)
\end{align*}
	Now we can describe a $\dMA^{\pub}$ protocol for $\mathcal{L}$: Merlin sends Arthur the proof $m$ that he would send in protocol $\mathcal{P}$. Then, Arthur proceeds to draw a random integer $i \in [t]$. Then, all nodes take their portion of $r_i$ and procede with the protocol $ \mathcal{P}$ using proof $m$ and $r_i$.  The completeness and soundness of the protocol is guaranteed by the choice of the set $\{r_i\}_{i=1}^t$. The total bandwidth of the protocol is $f(n) + \log t=f(n) + \cO( \log (n) + \log ( \delta^{-1}) + \log (f(n)) ) $ bits.
	We deduce that $\mathcal{L}$ belongs to $\dMA^{\pub}[f(n) + \log n + \log (\delta^{-1})]$.

\end{proof}

As we did in previous section for \dAM\ protocols, we are going to give here a negative answer to the question whether $\dMA^{\pub}$ and $\dMA^{\priv}$ are equivalent models.
For obtaining such separation, we use the problem $\twocoleq$. Recall that this language is the set of network configurations $(G,\id,I)$, where 
$I$ is a function $I: V(G) \rightarrow \{0,1\}^n$, such that $I$ is a proper two-coloring of $G$. In other words, $(G,\id,I)$ belongs to $\twocoleq$ if and only if there is a partition $\{V_0, V_1\}$ of $V(G)$, such that both $V_0$ and $V_1$ are inependent sets and, for all $v,w \in V_i$, we have that $I(v) = I(w)$, for $i \in \{0,1\}$.

Next lemma shows that \twocoleq\ is ``easy" to solve using shared randomness.

\begin{lemma}\label{twoColEqPub}
	$\twocoleq \in \dMA^{\pub}[\log n]$. 

\end{lemma}

\begin{proof}

	The protocol is the following. First, the prover sends a single bit $c_v\in \{0,1\}$ to each node $v$, that corrresponds to the 2-coloring. 
	Then, each node considers the smallest prime $q$ such that $n^{c+2} \leq q \leq2n^{c+2}$  and constructs a polynomial 
	over the field $\mathbb{F}_q$ associated to its input $I(v)$,  given by $p_v(z) = \sum_{i=1}^n I(v)_i \:z^i$. Finally, during the communication round, all nodes generate a random string $s \in \mathbb{F}_q$ using the shared randomness, and communicate $p_v(s)$. From this exchange, each node locally verifies the consistency of the $2$-coloring and that $p_u(s)$ equals $p_w(s)$,  for every pair of neighbors $u,w$. A node accepts if both conditions are satisfied and rejects otherwise. The bandwidth of the protocol is $\cO(\log n)$ bits.
	
	\begin{itemize}
		\item \textbf{{Completeness: }} If $(G,\id,I)$ is a yes-instance of $\twocoleq$, then the input graph is bipartite. The colors  $\{c_v\}_{v\in V}$  
		and $\{I(v)\}_{v \in V}$  induce the same bipartition. Obviously, in that case, every couple of neighbors $u,w$ of $v$ satisfy
		$p_u(s) = p_w(s)$. Therefore, the nodes always accept.
		
		\item \textbf{{Soundness: }} If $G$ is not bipartite the nodes immediately reject because in that case the coloring $c$ is not consistent (i.e. $c_v = c_u$ for two adjacent vertices $u$, $v$). Suppose now that $G$ is bipartite with partition $\{V_0,V_1\}$, and suppose without loss of generality that there are two vertices $u, v$ in $V_0$ such that $I(u )\neq I(v)$. Observe that, since $G$ is connected, we can choose $u,v$ with a common neighbor $w \in V_1$. Then, $w$ receives $p_u(s)$ and $p_v(s)$ in the verification round. The probability that $p_u(s) = p_v(s)$ is at most $n/q$. Indeed,  $p(z) = p_u(z)-p_v(z)$ is a polynomial of degree at most $n$ in $\mathbb{F}_q$, and, then, it has  at most $n$ roots in $\mathbb{F}_q$. We conclude that $w$ accepts with probability at most $1/n^{c+1}$.  
		
	\end{itemize}
	
	We deduce that $\twocoleq$ belongs to $\dMA^{\pub}[\log n]$. 
\end{proof}

The goal now is to prove that $\twocoleq \in \dMA^{\priv}[\Theta(\sqrt{n})]$.
We divide this proof  in two subsections: the first for the upper bound  and the second for the lower bound.

\subsection{The upper bound}

Babai and Kimmel devise a private coin, randomized protocol in the simultaneous messages model ($\mathsf{SM}$)  that solves $\eq$ communicating 
$\cO(\sqrt{n})$ bits \cite{babai1997randomized}. Problem $\eq$ consists in deciding whether two $n$-bit boolean vectors,  the inputs of Alice and Bob,  are equal. 

\begin{proposition}[\cite{babai1997randomized}]\label{EQsym}
	There exists a pribate coin, randomized protocol in the $\mathsf{SM}${} model that solves  \eq\  using $\cO(\sqrt{n})$ bits.
\end{proposition}
By using the protocol of Babai and Kimel,  we can directly construct a $\dMA^{\priv}$ protocol for $\twocoleq$.

\begin{lemma} 
	 $\twocoleq \in \dMA^{\priv}[\sqrt{n}]$.
\end{lemma}

\begin{proof}
	The prover sends each node $v$ the bit $c_v$ that defines the $2$-coloring and then the nodes proceed to broadcast a message according to the protocol in \cref{EQsym}. 
	Then, they locally verify the consistency of the $2$-coloring and each node $w$ in $V_i$ proceeds to act as referee for each pair of nodes $u,v$ in its vicinity, accepting if for each pair of nodes the referee would accept.
	\begin{itemize}
		\item \completeness\ If the input corresponds to a \emph{yes}-instance, then the nodes always accept: they receive and verify the  $2$-coloring and for all pairs of neighbors check that the equality protocol holds, as the protocol from \cref{EQsym} has one sided error, the graph accepts with perfect probability.
		
		\item\soundness\ Suppose that the input corresponds to a \emph{no}-instance. If $G$ is not bipartite the nodes immediately reject because in that case the coloring $c$ is not consistent (i.e. $c_v = c_u$ for two adjacent vertices $u$, $v$). Suppose now that $G$ is bipartite with partitions $V_0$ and $V_1$, and suppose without loss of generality that there are two vertices $u, v$ in $V_0$ such that $I_u \neq I_v$. Observe that, since $G$ is connected, we can choose $u,v$ with a common neighbor $w \in V_1$. Then,  the probability that $w$ accepts is at most $\eps$, where $\eps$ is the acceptation error of the protocol described in \cref{EQsym}.
		\end{itemize}\end{proof}

\subsection{The lower bound}

In order to give a lower-bound on the bandwidth of any $\dMA^{\priv}$ protocol solving $\twocoleq$, we show that the result of Babai and Kimmel given by \cref{babai-kimmel} can be extended to the scenario where Alice and Bob have access to random bits.

\begin{theorem}\label{lowerMAsym}

Let $f: X \times Y \to \{0,1\}$ be any boolean function. Let $0<\eps<\frac{1}{2}$.  Any $\eps$-error $\MA^{\sym}$ protocol for  solving 
	$f$ using private coins needs the messages to be of size
	at least $\Omega\left(\sqrt{\mathsf{M}^{\sym}(f)}\right)$.
\end{theorem}

\begin{proof}
	Let $f: X \times Y  \to \{0,1\}$ be a boolean function and consider $\mathcal{P}$ to be a $ \MA^{\sym}$ protocol with two sided error $\eps$, where the size of the messages sent by Alice and Bob and the size of the proof are bounded by $K$.
	
	Let $\Gamma$ be the set of all possible proofs sent by Merlin.
	Let $\Omega$ and $\Phi$ to be the set of all possible messages sent by Alice and Bob, with $a$ and $b$ bits. Now, given input $x$ and proof $m$, we define $\mu_{x,m}$ to be distribution of messages sent by Alice given her input and the proof received. We define $\nu_{y,m}$ analogously for Bob.
	
	Now, set $T= \{w_i\}_{i=1}^t$ to be a multiset of elements in $\Omega$  obtained uniformly at random, with \(\mu(T) = \sum_{i=1}^t \mu(w_i)\)
	and define $\rho(\omega, \varphi)$ as be the indicator function of whether the referee accepts given messages $\omega$ and $\varphi$. 
	
	We then have that $\mathcal{P}$'s correctness can be restated as follows:
	\begin{align*}  &f(x,y)=1 \longrightarrow \exists m,\quad   \underset{\omega, \varphi}{\sum} \mu_{x,m} (\omega) \cdot \nu_{y,m}(\varphi) \rho ( \omega, \varphi) \geq 1-\eps \\
	&f(x,y) = 0 \longrightarrow \forall m, \quad \underset{\omega, \varphi}{\sum}\mu_{x,m} (\omega) \cdot \nu_{y,m}(\varphi) \rho ( \omega, \varphi) \leq \eps \end{align*}
	Finally, consider the  \textit{strength} of $\varphi$ over $x$, given the proof $m$ to be defined as
	\[ F( x, \varphi, m ) =\sum_{\omega \in \Omega } \mu_{x,m}(\omega)\cdot \rho(\omega, \varphi) \]
	And, given an input $x\in X$, a proof $m$, $\varphi$ and a multiset $T= \{w_i\}_{i=1}^t$ we set the variables 
	\[ \xi_i(\varphi, m ) = \begin{cases} 1 &\text{ if the referee accepts } (\omega_i, \varphi) \text{ given }m\\
	0 &\text{ if not}
	\end{cases}\]
	\begin{remark}\label{expectedXi}\label{chernoff}
		The variables $\xi_i(\varphi, m)$ are independent and their expected value is $F(x, \varphi, m)$
	\end{remark}

	\begin{claim}\label{multisets}
		For all input $x$  and proof $m$ there exists a multiset $T_{x,m} = \{w_1, \dots w_t\}$ with $t = O( \log (|\Phi| ))$ such that for any $\varphi \in \Phi$.
		\[ \biggl |\sum_{i=1}^{t}\xi_i(\varphi, m )- t \cdot F(x,\varphi, m)\biggr| \leq \delta\cdot t  \]
	\end{claim}
	\begin{claimproof}
		Indeed, if we define $\Lambda( \varphi, m, T)$ to be the event $\left\{ |\sum_{i=1}^{t}\xi_i(\varphi, m )- t \cdot F(x,\varphi, m)| > \delta\cdot t \right\}$ choosing $T$ uniformly at random. Following \cref{chernoff} we obtain by a Chernoff bound that:
		\[ \mathbf{Pr}_{ T} \left(\Lambda(\varphi, m, T_x)\right) < 2 \cdot e^{-(\delta\cdot t)^2 /2t} = 2e^{-t\delta^2/2} < \frac{1}{2|\Phi|}\]
		By taking a large enough constant for $t$. Then
		\[ 	\mathbf{Pr}_{ T}(  \exists \varphi \text{ s.t. } \Lambda(\varphi, m, T_x) ) < 1/2	\]
		And so, by the probabilistic method, there exists a $T_{x,m}$ such that for any $\varphi$ we have
		\[  	\biggl|\sum_{i=1}^{t}\rho_m(\omega_i, \varphi )- t \cdot F(x,\varphi, m)\biggr| \leq \delta\cdot t 	\]
	\end{claimproof}
	We construct a collection $\{T_{y,m}\}_{(y,m) \in Y\times\Gamma}$ for each input $y$ for Bob and each proof $m$ in a similar way .
	\begin{claim}
		For any $x,y$ and proof $m$ the pair $(T_{x,m},T_{y,m})$ induces a non deterministic protocol for $f$.
	\end{claim}
\begin{claimproof}
	We may first assume, without loss of generality, that the referee's decision is deterministic: for any tuple $(\omega, \varphi, m)$ the referee outputs the most probable answer over his random bits, duplicating the error \cite{newman1996public}.  And so we may consider $\rho_m(\omega, \varphi)$ to be the indicator function over the referee's decision given messages $(\omega, \varphi)$  and the proof $m$.
	
	For $x \in X$, consider $T_{m,x} = (\omega_1, \dots, \omega_t)$ and for $y \in Y$, $T_{m,y} = (\varphi_1, \dots, \varphi_t)$ the collection of messages obtained by \cref{chernoff}.
	
Also, we consider the acceptance probability of the pair $x,y$ given a proof $m$ as
	\[F(x,y,m) = \underset{\omega, \varphi}{\sum}\mu_{x,m} (\omega)  \ \nu_{y,m}(\varphi) \cdot\rho_m ( \omega, \varphi)\]
	By the definition of $T_{x,m}$ and $T_{y,m}$ we have that:
	\[ \biggl |\sum_{i=1}^{t}\rho_m(\omega_i, \varphi)- t \cdot F(x,\varphi, m)\biggr | \leq \delta\cdot t  \text{\: and\: } \biggl |\sum_{j=1}^{t}\rho_m(\omega, \varphi_j)- t \cdot F(\omega, y, m)\biggr | \leq \delta\cdot t\]
	with the strength of $\varphi$ and $\omega$ with respect to $x$ and $y$ being defined in the same way as before. And so

	\begin{align*} &\underset{\omega}{\sum}\mu_{x,m} (\omega) \rho_m ( \omega, \varphi) \leq \frac{1}{t} \sum_{i=1}^{t}\rho_m(\omega_i, \varphi) + \delta \\
	&\underset{\varphi}{\sum}\nu_{y,m} (\varphi) \rho_m ( \omega, \varphi) \leq \frac{1}{t}\sum_{j=1}^{t}\rho_m(\omega, \varphi_j)+ \delta
	\end{align*}
	
this allows to bound the acceptance probability of $x$ and $y$ as:
	\begin{align*}
	F(x,y,m)  
	&\leq   \underset{\omega}{\sum}\mu_{x,m} (\omega) \biggl( \frac{1}{t}\sum_{j=1}^{t}\rho_m(\omega, \varphi_j)+ \delta  \biggr)  \\
	&\leq  \delta +  \biggl(\underset{\omega}{\sum}\mu_{x,m} (\omega)  \frac{1}{t}\sum_{j=1}^{t}\rho_m(\omega, \varphi_j) \biggr) \\
	&  \leq   \delta + \frac{1}{t}\sum_{j=1}^{t} \biggl(\underset{\omega}{\sum}\mu_{x,m} (\omega) \rho_m(\omega, \varphi_j) \biggr)   \\
	&\leq   2\delta + \frac{1}{t^2 }\sum_{i,j=1}^{t} \rho_m(\omega_i, \varphi_j)  
	\end{align*}
by replicating the above procedure for the other direction we obtain:
	\[ \biggl|\sum_{i,j=1}^{t}\rho_m(\omega_i, \varphi_j)- t^2 \cdot F(x,y, m)\biggr| \leq 2\delta t^2 \]
	In other words, if the referee receives $T_{x,m}$ and $T_{y,m}$ he may estimate the value of $F(x,y,m)$ by a factor of $2\delta$ and accept or reject accordingly.
	
	From here we can define the protocol $\mathcal{P}^*$ simply as follows: Alice and Bob send $T_{x,m}$ and $T_{y,m}$ respectively. Then the referee takes de average answer for all pairs $(\omega_i, \varphi_j)$ given $m$ and accepts if the majority of the cases accept.
	\begin{itemize}
		
		\item \completeness: If $(x,y)$ is a yes-instance, then there exists a proof $m$ such that $F(x,y,m)$ is large ($\geq1-\eps$) . As we know that
		$ \frac{1}{t^2}|\sum_{i,j=1}^{t}\rho_m(\omega_i, \varphi_j)- t^2 \cdot F(x,y, m)| \leq 2\delta $ we have that $ \frac{1}{t^2}\sum_{i,j=1}^{t}\rho_m(\omega_i, \varphi_j) \geq 1-\eps-2\delta$. Therefore by choosing $\delta$ sufficiently small the referee accepts $T_{x,m}$ and $T_{y,m}$.
		
		\item \soundness: If $(x,y)$ is a no-instance, then for any proof $m$ if follows that $F(x,y,m)$ is small ($\leq \eps$). And as we know that
		$ \frac{1}{t^2}|\sum_{i,j=1}^{t}\rho_m(\omega_i, \varphi_j)- t^2 \cdot F(x,y, m)| \leq 2\delta $ then we have that $ \frac{1}{t^2}\sum_{i,j=1}^{t}\rho_m(\omega_i, \varphi_j) \leq \eps+ 2\delta$. And so the referee rejects $T_{x,m}$ and $T_{y,m}$ for any $m$.
	\end{itemize}\end{claimproof}
Thus given an $\MA^{\sym}$ protocol for $f$ using $O(K)$ bits we obtained a $\M^{\sym}$ protocol that uses $O(K^2)$ bits. Therefore $K = \Omega\left( \sqrt{\textsf{M}^{\sym}(f)}\right)$.
\end{proof}

\begin{lemma}\label{2colEQ}
	If $\twocoleq \in \dMA_{\eps}^{\priv}[f(n)]$ with $\eps <1/4$, then there exists a protocol $\mathcal{P}$ solving $\textsc{Equality}$ in the $\mathsf{MA}^{sym}$ model with bandwidth $\cO(f(n))$.
\end{lemma}

\begin{proof}
	Indeed, let $\mathcal{P}$ be a protocol for $\twocoleq$ in the model $\dMA$\ using random coins. We design a protocol $\mathcal{P}^*$ in the $\MA^{\sym}$ defined as follows. Let $x,y \in \{0,1\}^n$, and assume without loss of generality that $n$ is even.  Given $n \in \mathbb{N}$ Alice, Bob and the referee  construct the following network configuration $(G, \id, I)$:
	
	\begin{itemize}
		\item $G$ is a path of $2n+1$ nodes $v_1, \dots, v_{2n+1}$.
		\item $\id(v_i) = i$ for each $i\in \{1, \dots, 2n+1\}$.
		\item 
		$I(v_i) = \left\{ \begin{array}{cl} 
		0^n & \textrm{if } i \textrm{ is odd } \\ 
		x & \textrm{if } i \textrm{ is even and } i \leq n\\
		y & \textrm{if } i \textrm{ is even and } i > n\\
		\end{array}\right.$

	\end{itemize} 

	\begin{figure}[!ht]
		\centering

\tikzset{every picture/.style={line width=0.75pt}} 
		
		\begin{tikzpicture}[x=0.75pt,y=0.75pt,yscale=-1,xscale=1]
		
		\draw    (240,140) -- (280,140) ;
		\draw    (120,140) -- (160,140) ;
		\draw    (160,140) -- (200,140) ;
		\draw    (200,140) -- (240,140) ;
		\draw  [fill={rgb, 255:red, 255; green, 255; blue, 255 }  ,fill opacity=1 ] (105,140) .. controls (105,131.72) and (111.72,125) .. (120,125) .. controls (128.28,125) and (135,131.72) .. (135,140) .. controls (135,148.28) and (128.28,155) .. (120,155) .. controls (111.72,155) and (105,148.28) .. (105,140) -- cycle ;
		\draw  [fill={rgb, 255:red, 255; green, 255; blue, 255 }  ,fill opacity=1 ] (185,140) .. controls (185,131.72) and (191.72,125) .. (200,125) .. controls (208.28,125) and (215,131.72) .. (215,140) .. controls (215,148.28) and (208.28,155) .. (200,155) .. controls (191.72,155) and (185,148.28) .. (185,140) -- cycle ;
		\draw  [fill={rgb, 255:red, 255; green, 255; blue, 255 }  ,fill opacity=1 ] (225,140) .. controls (225,131.72) and (231.72,125) .. (240,125) .. controls (248.28,125) and (255,131.72) .. (255,140) .. controls (255,148.28) and (248.28,155) .. (240,155) .. controls (231.72,155) and (225,148.28) .. (225,140) -- cycle ;
		\draw  [fill={rgb, 255:red, 255; green, 255; blue, 255 }  ,fill opacity=1 ] (145,140) .. controls (145,131.72) and (151.72,125) .. (160,125) .. controls (168.28,125) and (175,131.72) .. (175,140) .. controls (175,148.28) and (168.28,155) .. (160,155) .. controls (151.72,155) and (145,148.28) .. (145,140) -- cycle ;
		\draw    (440,140) -- (480,140) ;
		\draw    (280,140) -- (340,140) ;
		\draw    (400,140) -- (440,140) ;
		\draw  [fill={rgb, 255:red, 255; green, 255; blue, 255 }  ,fill opacity=1 ] (265,140) .. controls (265,131.72) and (271.72,125) .. (280,125) .. controls (288.28,125) and (295,131.72) .. (295,140) .. controls (295,148.28) and (288.28,155) .. (280,155) .. controls (271.72,155) and (265,148.28) .. (265,140) -- cycle ;
		\draw    (340,140) -- (400,140) ;
		\draw  [fill={rgb, 255:red, 255; green, 255; blue, 255 }  ,fill opacity=1 ] (325,140) .. controls (325,131.72) and (331.72,125) .. (340,125) .. controls (348.28,125) and (355,131.72) .. (355,140) .. controls (355,148.28) and (348.28,155) .. (340,155) .. controls (331.72,155) and (325,148.28) .. (325,140) -- cycle ;
		\draw  [fill={rgb, 255:red, 255; green, 255; blue, 255 }  ,fill opacity=1 ] (425,140) .. controls (425,131.72) and (431.72,125) .. (440,125) .. controls (448.28,125) and (455,131.72) .. (455,140) .. controls (455,148.28) and (448.28,155) .. (440,155) .. controls (431.72,155) and (425,148.28) .. (425,140) -- cycle ;
		\draw    (480,140) -- (520,140) ;
		\draw    (520,140) -- (560,140) ;
		\draw  [fill={rgb, 255:red, 255; green, 255; blue, 255 }  ,fill opacity=1 ] (385,140) .. controls (385,131.72) and (391.72,125) .. (400,125) .. controls (408.28,125) and (415,131.72) .. (415,140) .. controls (415,148.28) and (408.28,155) .. (400,155) .. controls (391.72,155) and (385,148.28) .. (385,140) -- cycle ;
		\draw  [fill={rgb, 255:red, 255; green, 255; blue, 255 }  ,fill opacity=1 ] (505,140) .. controls (505,131.72) and (511.72,125) .. (520,125) .. controls (528.28,125) and (535,131.72) .. (535,140) .. controls (535,148.28) and (528.28,155) .. (520,155) .. controls (511.72,155) and (505,148.28) .. (505,140) -- cycle ;
		\draw  [fill={rgb, 255:red, 255; green, 255; blue, 255 }  ,fill opacity=1 ] (545,140) .. controls (545,131.72) and (551.72,125) .. (560,125) .. controls (568.28,125) and (575,131.72) .. (575,140) .. controls (575,148.28) and (568.28,155) .. (560,155) .. controls (551.72,155) and (545,148.28) .. (545,140) -- cycle ;
		\draw  [fill={rgb, 255:red, 255; green, 255; blue, 255 }  ,fill opacity=1 ] (465,140) .. controls (465,131.72) and (471.72,125) .. (480,125) .. controls (488.28,125) and (495,131.72) .. (495,140) .. controls (495,148.28) and (488.28,155) .. (480,155) .. controls (471.72,155) and (465,148.28) .. (465,140) -- cycle ;
		\draw  [color={rgb, 255:red, 74; green, 144; blue, 226 }  ,draw opacity=1 ][dash pattern={on 4.5pt off 4.5pt}] (100,118) .. controls (100,108.06) and (108.06,100) .. (118,100) -- (282,100) .. controls (291.94,100) and (300,108.06) .. (300,118) -- (300,172) .. controls (300,181.94) and (291.94,190) .. (282,190) -- (118,190) .. controls (108.06,190) and (100,181.94) .. (100,172) -- cycle ;
		\draw  [color={rgb, 255:red, 208; green, 2; blue, 27 }  ,draw opacity=1 ][dash pattern={on 4.5pt off 4.5pt}] (380,118) .. controls (380,108.06) and (388.06,100) .. (398,100) -- (562,100) .. controls (571.94,100) and (580,108.06) .. (580,118) -- (580,172) .. controls (580,181.94) and (571.94,190) .. (562,190) -- (398,190) .. controls (388.06,190) and (380,181.94) .. (380,172) -- cycle ;
		\draw  [color={rgb, 255:red, 245; green, 166; blue, 35 }  ,draw opacity=1 ][dash pattern={on 4.5pt off 4.5pt}] (310,112) .. controls (310,105.37) and (315.37,100) .. (322,100) -- (358,100) .. controls (364.63,100) and (370,105.37) .. (370,112) -- (370,178) .. controls (370,184.63) and (364.63,190) .. (358,190) -- (322,190) .. controls (315.37,190) and (310,184.63) .. (310,178) -- cycle ;
		
		\draw (113,138) node [anchor=north west][inner sep=0.75pt]  [font=\footnotesize]  {$v_{1}$};
		\draw (153,138) node [anchor=north west][inner sep=0.75pt]  [font=\footnotesize]  {$v_{2}$};
		\draw (190,142) node [anchor=north west][inner sep=0.75pt]    {$\dotsc $};
		\draw (226,138) node [anchor=north west][inner sep=0.75pt]  [font=\footnotesize]  {$v_{n-2}$};
		\draw (266,138) node [anchor=north west][inner sep=0.75pt]  [font=\footnotesize]  {$v_{n-1}$};
		\draw (333,138) node [anchor=north west][inner sep=0.75pt]  [font=\footnotesize]  {$v_{n}$};
		\draw (386,138) node [anchor=north west][inner sep=0.75pt]  [font=\footnotesize]  {$v_{n+1}$};
		\draw (425,138) node [anchor=north west][inner sep=0.75pt]  [font=\footnotesize]  {$v_{n+2}$};
		\draw (470,142) node [anchor=north west][inner sep=0.75pt]  [font=\normalsize]  {$\dotsc $};
		\draw (504,138) node [anchor=north west][inner sep=0.75pt]  [font=\footnotesize]  {$v_{2n-1}$};
		\draw (550,138) node [anchor=north west][inner sep=0.75pt]  [font=\footnotesize]  {$v_{2n}$};
		\draw (111,158) node [anchor=north west][inner sep=0.75pt]    {$0^{n}$};
		\draw (231,158) node [anchor=north west][inner sep=0.75pt]    {$0^{n}$};
		\draw (431,158) node [anchor=north west][inner sep=0.75pt]    {$0^{n}$};
		\draw (511,158) node [anchor=north west][inner sep=0.75pt]    {$0^{n}$};
		\draw (331,158) node [anchor=north west][inner sep=0.75pt]    {$0^{n}$};
		\draw (154,158) node [anchor=north west][inner sep=0.75pt]    {$x^{\phantom{n}}$};
		\draw (273,158) node [anchor=north west][inner sep=0.75pt]    {$x^{\phantom{n}}$};
		\draw (395,158) node [anchor=north west][inner sep=0.75pt]    {$y^{\phantom{n}}$};
		\draw (556,158) node [anchor=north west][inner sep=0.75pt]    {$y^{\phantom{n}}$};

		\end{tikzpicture}
		\caption{An instance $(G,\id, I)$ constructed by Alice and Bob: the blue box corresponds to the set of nodes assigned to Alice, along with input $x$, those in the red box are the ones assigned to Bob, along with the input $y$ while the orange box containing a single node is assigned to the referee, whose input is fixed.}
		\label{fig:bipEQ}
	\end{figure}
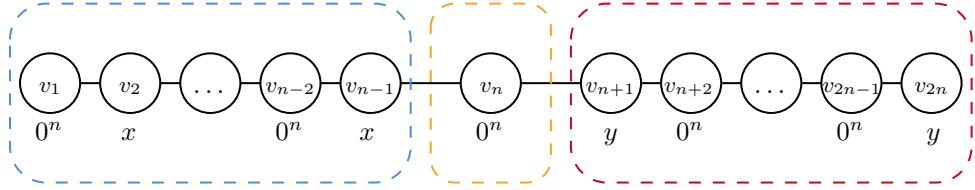

	Given the input $x$ for Alice and $y$ for Bob, they proceed to construct the instance $(G,\id,I)$: Alice takes the first $n$ nodes of $G$ while Bob takes the last $n$. Finally, the central node is assigned to the referee.
	
	For each $v\in G$, let $m(v)$ be the certificate that Melin sends to node $v$ according to protocol $\mathcal{P}$. In protocol $\mathcal{P}^*$, Alice receives from the prover the certificate $(m(v_n), m(v_{n+1}))$, and Bob receives the certificate $(m(v_{n+1}), m(v_{n+2}))$. Then, Alice proceeds to enumerate all possible certificates for the nodes $v_1, \dots, v_{n+1}$, along with all possible random messages that these vertices may generate, depending on each certificate and input. Likewise, Bob enumerates all possible certificates and random messages for nodes $v_{n+1}, \dots, v_{2n+1}$.
	
	Having simulated all possible interactions between the nodes in their section and the referee, Alice draws a random string $r_1$  and communicate the following to the referee the message $(\acc_A, s_n, m(v_{n+1}))$, where:
	\begin{itemize}
	\item $\acc_A\in \{0,1\}$ and equals $1$ if and only if all the vertices $v_1, \dots, v_n$ accept in protocol $\mathcal{P}$ for more than $1-\eps$ of all possible random bits; 	
	\item  $s_n$ is the message that node $v_n$ sends  to $v_{n+1}$ in the verification round of protocol $\mathcal{P}$, given the string $r_1$, the input $x$ and the certificate $m(v_n)$.
	\end{itemize}
	Analogously, Bob draws a random string $r_2$ and sends the ref the message $(\acc_B, s_{n+2}, m(v_{n+1}))$, such that  $\acc_B\in \{0,1\}$ and equals $1$ if and only if all the vertices $v_{n+2}, \dots, v_{2n+1}$ accept in protocol $\mathcal{P}$ for more than $1-\eps$ of all possible random bits; 	and $s_{n+2}$ is the message that $v_{n+2}$  sends to $v_{n+1}$ in the verification round of $\mathcal{P}$, given $r_2$, the input $y$ and the certificate $m(v_{n+2})$. We have that both Alice and Bob send $\cO(K)$ bits each.
	
	Having these messages, the referee verifies that $\acc_A = \acc_B = 1$ and that both send the same certificate for node $v_{n+1}$, rejecting if any of these fails. Then the referee draws a random string $r_3$ and simulates the verification round between the central node $v_{n+1}$, and nodes $v_{n}$ and $v_{n+2}$. Given the messages received from Alice and Bob, the referee has the messages that $v_{n+1}$ receives in the verification round of $\mathcal{P}$. Finally, the referee accepts if $v_n$ accepts. We now analyze the soundness and completeness of $\mathcal{P}^*$.
	\begin{itemize}
		\item \completeness If $x=y$ then $(G, \id, I)$ is a \emph{yes}-instance of $\twocoleq$. By the completeness of $\mathcal{P}$, all the nodes in $G$ accept with probability greater than $1-\eps$. This implies that Alice and Bob communicate $\acc_A = \acc_B = 1$ to the referee. It also implies that $v_{n+1}$ accepts. Therefore, the referee accepts with probability greater than $1-\eps$.
		
		\item \soundness In the case that $x\neq y$, we have that, by the correctness of $\mathcal{P}$, the probability that all nodes accept is strictly less than $\eps$. 
		
		Now consider $\textsf{A}$ to be the variable that equals 1 if Alice's portion of the graph accepts when $v_n$ receives $r_3$ from $v_{n+1}$ and $\textsf{B}$ be the variable that equals 1 if Bob's portion accepts given $r_3$. As $\acc_A = \acc_B = 1$ we have that

		\begin{align*}
		\textbf{Pr}( \text{The referee accepts}) &=\ \mathbf{Pr}(\textsf{A}, \textsf{B} \text{ and the referee accept}) \\ &+\ \mathbf{Pr}(\text{ The referee accepts and } \textsf{A}= 0 \text{ or } \textsf{B} = 0)\\
		&< \eps + 2 \eps
		\end{align*}
			
		As Alice and Bob each reject the protocol with probability less than  $\eps$. Therefore, with probability $1 -3\eps>\frac34$ we have that the referee rejects.
	\end{itemize}
Finally, as there is a gap between both acceptance probabilities, the error can be reduced by standard amplification.	We conclude that $\mathcal{P}^*$ is a protocol for  $\textsc{Equality}$ in the $\MA^{\sym}$ model with bandwidth $\cO(f(n))$. 
\end{proof}

\begin{theorem}
$\twocoleq \in \dMA_{\eps}^{\priv}[\Theta(\sqrt{n})]$ for any $\eps <\frac14$ and $\twocoleq \in \dMA_{1/3}^{\pub}[\Theta(\log n)]$.
\end{theorem}

\begin{proof}
	Indeed, it is known that in the classic $2$-party communication model of Alice and Bob the problem $\eq$ has complexity $\Theta(n)$ even with the help of no-determinism \cite{kushilevitz1997communication}. This bound translates naturally to the simultaneous messages model, and so $\textsf{N}(\eq) = \Theta(n)$. From \cref{lowerMAsym} we deduce that any protocol in the model  $\textsf{MA}^{\textsf{sym}}$ for $\eq$ using random bits requires  $\Theta(\sqrt{n})$ bits. Now, set $\eps < 1/4$. If there exists a protocol $\mathcal{P}$ for \twocoleq\ using $o(\sqrt{n})$ bits and with and error smaller than $\eps$, then by \cref{2colEQ} there would exist a protocol $\mathcal{P}^*$ for $\eq$ in the model $\textsf{MA}^{\textsf{sym}}$ using $o(\sqrt{n})$ bits and error smaller than $1/3$, a contradiction.
	
Moreover, for every $\eps \leq 1/3$,  if $\twocoleq$ belongs to $\dMA_{\eps}^{\pub}[f(n)]$ then $f(n) = \Omega(\log n)$ as we can derandomize the protocol and it would contradict the bound for $\eq$.	Thus by \cref{twoColEqPub} we conclude that the protocol is tight.
\end{proof}


\section{Open Problems}

Besides the main questions regarding the actual power of  $\dAM$ and $\dMA$, and a general method to obtain lower-bounds on these models, this work leaves open several interesting research perspectives. First, Theorem \ref{thm:simula} shows that, if  $f = \Omega(\log n)$, then  $\dAM^{\pub}[f] \subseteq \dM[2^f]$. Is it possible to obtain such an inclusion even for $f=o(\log n)$? 
We were able to show that this is indeed the case when the  interactive proof is restricted to randomized protocols with low error
(see Theorem \ref{thm:whp}), but it is unclear whether this holds in general. 

A second natural question is about the maximum gap between $\dAM$ protocols with private and shared randomness. 
More precisely, is there a language  contained in both  $\dAM^{\pub}[\Omega(n)]$ and  $\dAM^{\priv}[\cO(1)]$?

Finally, we can consider another variant of the model, which combines the power of shared and private randomness on any round. Namely, a model where the nodes use private randomness to interact with the prover and shared randomness in the verification round. How powerful is this model regarding the one with only private coins?

\bibliography{i-ref}

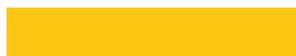
\end{document}